\documentclass[conference]{IEEEtran}
\usepackage[boxruled]{algorithm2e}
\usepackage{cite}
\usepackage{graphicx}
\usepackage{psfrag}
\usepackage{subfigure}
\usepackage{url}
\usepackage{amsmath}
\usepackage{array}
\usepackage{amssymb}
\usepackage{amsfonts}

\newtheorem{definition}{Definition}
\newtheorem{theorem}{Theorem}

\title{Two-User Gaussian Interference Channel with Finite Constellation Input and FDMA}
\begin{document}

\author{
\authorblockN{G. Abhinav}
\authorblockA{Dept. of ECE, Indian Institute of Science \\
Bangalore 560012, India\\
Email: abhig\_88@ece.iisc.ernet.in
}
\and
\authorblockN{B. Sundar Rajan}
\authorblockA{Dept. of ECE, Indian Institute of Science, \\Bangalore 560012, India\\
Email: bsrajan@ece.iisc.ernet.in
}
}

\maketitle
\thispagestyle{empty}	

\begin{abstract}
In the two-user Gaussian Strong Interference Channel (GSIC) with finite constellation inputs, it is known that relative rotation between the constellations of the two users enlarges the Constellation Constrained (CC) capacity region. In this paper, a metric for finding the approximate angle of rotation (with negligibly small error) to maximally enlarge the CC capacity for the two-user GSIC is presented. In the case of Gaussian input alphabets with equal powers for both the users and the modulus of both the cross-channel gains being equal to unity, it is known that the FDMA rate curve touches the capacity curve of the GSIC. It is shown that, with unequal powers for both the users also, when the modulus of one of the cross-channel gains being equal to one and the modulus of the other cross-channel gain being greater than or equal to one, the FDMA rate curve touches the capacity curve of the GSIC. On the contrary, it is shown that, under finite constellation inputs, with both the users using the same constellation, the FDMA rate curve strictly lies within (never touches) the enlarged CC capacity region throughout the strong-interference regime. This means that using FDMA it is impossible to go close to the CC capacity. It is well known that for the Gaussian input alphabets, the FDMA inner-bound, at the optimum sum-rate point, is always better than the simultaneous-decoding inner-bound throughout the weak-interference regime. For a portion of the weak interference regime, it is shown that with identical finite constellation inputs for both the users, the simultaneous-decoding inner-bound, enlarged by relative rotation between the constellations, is strictly better than the FDMA inner-bound. 

\end{abstract}	
\section{INTRODUCTION AND PRELIMINARIES}
\label{sec1}
The Gaussian Interference channel (\textit{GIC}) model \cite{GaY}, is shown in Fig \ref{fig:AWGNIC}. User-1 intends to communicate with Receiver-1 at rate $R_1$ and User-2 with Receiver-2 at rate $R_2$, with both the users interfering with each other at their respective receivers as dictated by the channel gains. Channel gain from User-$i$ to Receiver-$j$ is denoted by $h_{ij}$. The users are equipped with complex signal constellations ${\cal S}_1$ and ${\cal S}_2$ of cardinality $M_1$ and $M_2$, with average power constraints $P_1$ and $P_2$ respectively. Symbol level synchronization between the users is assumed. The signals obtained at the receivers are given by
\begin {align}
\label{eqnset1}
Y_1 = h_{11}X_1 + h_{21}X_2 + N_1\\
Y_2 = h_{12}X_1 + h_{22}X_2 + N_2
\end {align}
where $X_1$ $\in$ ${\cal{S}}_1$, $X_2$ $\in$ ${\cal S}_2$, $N_1$ $\sim$ ${\cal CN}(0,\sigma_1^2)$, $N_2$ $\sim$ ${\cal CN}(0,\sigma_2^2)$. (${\cal CN}(0,\sigma_j^2)$ represents circularly symmetric complex Gaussian noise with mean $0$ and variance $\sigma_j^2$, ($j$=$1,2$)).
\begin{figure}[htbp]
\centering
\includegraphics[totalheight=1.5in,width=2in]{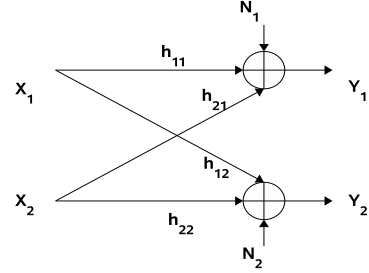}
\caption{GIC Model}	
\label{fig:AWGNIC}	
\end{figure}

Without loss of generality, throughout this paper, we assume $h_{11}$=$h_{22}$=$1$ 
\cite{Car}.

Define
\begin{align}
\nonumber
&SNR_1 = \frac{P_1}{\sigma_1^2},\\
\nonumber
&INR_2 = \frac{|h_{12}|^2P_1}{\sigma_2^2},\\
\nonumber
&SNR_2 = \frac{P_2}{\sigma_2^2},~ \mbox{and}\\
\nonumber
&INR_1 = \frac{|h_{21}|^2P_2}{\sigma_1^2}
\end{align} where, $SNR_i$ and $INR_i$ ($i$=$1,2$) denote the intended-signal to noise power ratio and the interference to noise power ratio at Receiver-$i$ respectively.

\begin{definition} [\cite{HS,CoE,GaY} ]\footnote{The conventional definition for strong interference treats the very strong interference as a special case; in this paper we exclude very strong interference from strong interference.}
\label{def1}
A GIC is said to be in strong interference when
\begin {align}
\nonumber
&SNR_1 \leq INR_2, \\
\label{eqnset3_a_1}
&SNR_2 \leq INR_1\\
\nonumber
&\hspace{-0.1cm} \mbox{and, atleast one of the following two conditions is satisfied}\\
\nonumber
&SNR_1 > \left (\frac{INR_2}{1+SNR_2} \right)  \mbox{,} \\
\label{eqnset3_a}
&SNR_2 > \left (\frac{INR_1}{1+SNR_1} \right).
\end {align}
\end{definition}

\begin{definition}\footnote{In the literature, different definitions for weak interference regime are available. In this paper, we stick to our definition.}
 A GIC is said to be in weak interference when atleast one of the conditions in (\ref{eqnset3_a_1}) is violated.
\end{definition}

For the two-user Gaussian strong interference channel (GSIC), the capacity region (in bits per channel use) is given by \cite{HS}
\begin {align}
\nonumber
R_1 &\leq log_2\left( 1 + \frac {P_1}{\sigma_1^2} \right)\\
\nonumber
R_2 &\leq log_2 \left( 1 + \frac {P_2}{\sigma_2^2} \right)\\
\label{eqnset5}
R_1 + R_2 &\leq min \left \{ ~log_2 \left( 1 + \frac{P_1 + |h_{21}|^2P_2}{\sigma_1^2} \right),\hspace{30cm}\right \}\\
\nonumber
&\hspace{-30cm}\left \{ \hspace{32cm} log_2 \left( 1 + \frac{|h_{12}|^2P_1 + P_2}{\sigma_2^2} \right) \right \}.
\end{align} 
Gaussian codebooks achieve the capacity in the GSIC. Though this capacity region provides insights into the achievable rate pairs ($R_1, R_2$) in an information theoretic sense, it fails to provide insight on the achievable rate pairs when we consider finitary restrictions on the input alphabets and analyze some real world practical signal constellations like QAM and PSK etc.

In this work we assume, that the two \textit{independent} users use finite complex constellations with uniform distribution over its elements. Under the above assumptions, the maximum achievable rate is referred to as the Constellation Constrained (CC) capacity \cite{Big}. The CC capacity was analyzed for the Gaussian-MAC (G-MAC) in \cite{HaR} and for the broadcast channel in \cite{DeR}. Recently, we came to know of the work on the CC capacity for the GSIC in \cite{KnS} in which capacity maximization for the GSIC by rotation of signal set is studied and it has been shown that only relative angle of rotation between the constellations matter. The optimum angle of rotation was computed numerically in \cite{KnS}. 

The contributions of this paper are as follows:
\begin{itemize}
\item We present a metric to obtain the approximate angle of rotation (with negligibly small error) required for maximal enlargement of the CC capacity region for the two-user GSIC that can be computed with considerable ease.
\item When the User-Receiver pair use the Frequency Division Multiple Access (FDMA) scheme, it is known that the rate curve when Gaussian alphabets are used, with $P_1$=$P_2$, touches the capacity curve of the GSIC when $SNR_1$=$SNR_2$=$INR_1$=$INR_2$ \cite{GaY}. We show that, (for the Gaussian alphabet case), with $P_1$ not necessarily equal to $P_2$, the FDMA rate curve touches the capacity curve of the GSIC, also when $INR_2$$\geq$$ SNR_1$ and $INR_1$=$SNR_2$ and when $INR_2$=$SNR_1$ and $INR_1$$\geq$$ SNR_2$. On the contrary, in the finite constellation case, with ${\cal S}_1$=${\cal S}_2$, we show that the FDMA rate curve always lies strictly inside (never touches) the CC capacity region of the GSIC.

\item It is known that, with $P_1$=$P_2$ and $INR_2$=$INR_1<SNR_1$=$SNR_2$, for the Gaussian alphabet case, the FDMA inner-bound, at the optimum sum-rate point, is better than the simultaneous-decoding\footnote{Throughout this paper, the simultaneous-decoding we refer to is the version of simultaneous-decoding that doesn't require the message of each user to be correctly decoded at the unintended receiver, as mentioned in \cite{GaY}.} inner-bound \cite{GaY}. We show that, with $P_1$ not necessarily equal to $P_2$, throughout the weak-interference regime, the FDMA inner-bound, at the optimum sum-rate point, is better than the simultaneous-decoding inner-bound for the Gaussian input, whereas, for the finite constellation case, with ${\cal S}_1$=${\cal S}_2$, for some portion of the weak interference regime, the simultaneous-decoding inner-bound is strictly better than the FDMA inner-bound. \end{itemize}

\textit{\textbf{Notations:}} For a random variable $X$ which takes value from the set $\cal S$, we assume some ordering of its elements and use $x^i$ to represent the $i$-th element of $\cal S$. Realization of the random variable X is denoted as $x$. Absolute value of a complex number $x$ is denoted by $|x|$ and $E[X]$ denotes the expectation of the random variable $X$. 
All the logarithms in this paper are evaluated for base-2.
\section{A METRIC FOR MAXIMAL CAPACITY ENLARGEMENT}
\label{sec2}
Throughout this section we consider two-user GSIC. The CC capacity for the GSIC, is given by \cite{HaK}
\begin {align}
\label{eqnset4}
\nonumber
R_1 &\leq I(X_1;Y_1|X_2)\\
\nonumber
R_2 & \leq I(X_2;Y_2|X_1)\\
R_1 + R_2 & \leq min \lbrace I(X_1,X_2;Y_1), I(X_1,X_2;Y_2)\rbrace.
\end{align} The above mutual informations can be easily evaluated as in (\cite{Big,HaR}) and are shown in (\ref{eqnset6_1}), (\ref{eqnset6_2}), (\ref{eqnset6_3}), and (\ref{eqnset6_5}) (at the top of the next page).

\begin{figure*}
\scriptsize
\begin{align}
\centering
 \label {eqnset6_1}
&I(X_1;Y_1|X_2)= logM_1 - \frac{1}{M_1} \sum_{k_1 = 0}^{M_1-1} E_{N_1} \left[log \left( \sum_{i_1=0}^{M_1-1} exp \left(- \left( \frac{|N_1 + \left( x_1^{k_1}-x_1^{i_1} \right)|^2-|N_1|^2}{\sigma_1^2} \right) \right)\right)\right]\\
\centering
\label {eqnset6_2}
&I(X_2;Y_2|X_1)= logM_2 - \frac{1}{M_2} \sum_{k_2 = 0}^{M_2-1} E_{N_2} \left[log \left( \sum_{i_2=0}^{M_2-1} exp \left(- \left( \frac{|N_2 + \left( x_2^{k_2}-x_2^{i_2} \right)|^2-|N_2|^2}{\sigma_2^2} \right) \right)\right)\right]\\
\centering
\label {eqnset6_3}
&I_1~\triangleq~I(X_1,X_2;Y_1)= log(M_1M_2) - \frac{1}{M_1M_2} \sum_{k_1 = 0}^{M_1-1} \sum_{k_2 = 0}^{M_2-1} E_{N_1} \left[log \left( \sum_{i_1=0}^{M_1-1} \sum_{i_2=0}^{M_2-1} exp \left(- \frac{|N_1 + \left( x_1^{k_1}-x_1^{i_1} \right) + h_{21}e^{j\theta} \left( x_2^{k_2}-x_2^{i_2} \right)|^2-|N_1|^2}{\sigma_1^2} \right)\right)\right]\\
\label {eqnset6_4}
&~~~~~~~~~~~~~~~~~~~\geq log(M_1M_2) - log~e - \frac{1}{M_1M_2} \sum_{k_1 = 0}^{M_1-1} \sum_{k_2 = 0}^{M_2-1} log \left( \frac{1}{2} \sum_{i_1=0}^{M_1-1} \sum_{i_2=0}^{M_2-1} exp \left(- \frac{| \left( x_1^{k_1}-x_1^{i_1} \right) + h_{21}e^{j\theta} \left( x_2^{k_2}-x_2^{i_2} \right)|^2}{2\sigma_1^2} \right)\right)~ \triangleq ~I_1 ^\prime\\
\label {eqnset6_5}
&I_2~\triangleq~I(X_1,X_2;Y_2)= log(M_1M_2) - \frac{1}{M_1M_2} \sum_{k_1 = 0}^{M_1-1} \sum_{k_2 = 0}^{M_2-1} E_{N_2} \left[log \left( \sum_{i_1=0}^{M_1-1} \sum_{i_2=0}^{M_2-1} exp \left(- \frac{|N_2 + h_{12} \left( x_1^{k_1}-x_1^{i_1} \right) + e^{j\theta} \left( x_2^{k_2}-x_2^{i_2} \right)|^2-|N_2|^2}{\sigma_2^2} \right)\right)\right]\\
\label {eqnset6_6}
&~~~~~~~~~~~~~~~~~~~~\geq log(M_1M_2) - log~e - \frac{1}{M_1M_2} \sum_{k_1 = 0}^{M_1-1} \sum_{k_2 = 0}^{M_2-1} log \left( \frac{1}{2} \sum_{i_1=0}^{M_1-1} \sum_{i_2=0}^{M_2-1} exp \left(- \frac{|h_{12} \left( x_1^{k_1}-x_1^{i_1} \right) + e^{j\theta} \left( x_2^{k_2}-x_2^{i_2} \right)|^2}{2\sigma_2^2} \right)\right)~ \triangleq ~I_2 ^\prime
\end{align}
\hrule
\end{figure*}

For channel gains taking complex values, since $N_1$ and $N_2$ are circularly symmetric Gaussian noise, rotation of either ${\cal S}_1$ or ${\cal S}_2$ by any arbitrary angle doesn't change the values in (\ref{eqnset6_1}) and (\ref{eqnset6_2}), where as the values in (\ref{eqnset6_3}) and (\ref{eqnset6_5}) do change. Hence, the CC capacity region does change, providing us with an option for maximally expanding it \cite{KnS}. Since, only relative angle of rotation between the constellations matter \cite{KnS}, we shall rotate only ${\cal S}_2$ and denote the angle of rotation as $\theta$.

Let ${\cal S}_{sum_1} = \lbrace x_1 + h_{21}x_2 | \forall x_1 \in {\cal S}_1, x_2 \in {\cal S}_2^\prime \rbrace$ and ${\cal S}_{sum_2} = \lbrace h_{12}x_1 + x_2 | \forall x_1 \in {\cal S}_1, x_2 \in {\cal S}_2^\prime \rbrace$, where ${\cal S}_2^\prime$ can be either an unrotated or a rotated version of ${\cal S}_2$. Define ${\varphi}_1 : {\cal S}_1 \times {\cal S}_2^\prime \longrightarrow {\cal S}_{sum_1}$ and ${\varphi}_2 : {\cal S}_1 \times {\cal S}_2^\prime \longrightarrow {\cal S}_{sum_2}$. The following theorem gives the metric for choosing an approximate angle of rotation to maximally enlarge the CC capacity region which, unlike in \cite{KnS}, doesn't involve numerical computation.
\begin{figure*}
\scriptsize
\begin{align}
\nonumber
& \theta_{opt} = arg~ \min_{\theta \in (0,2\pi)}~max \left\{ ~ \sum_{k_1 = 0}^{M_1-1} \sum_{k_2 = 0}^{M_2-1} log \left( \sum_{i_1=0}^{M_1-1} \sum_{i_2=0}^{M_2-1} exp \left(- \frac{| \left( x_1^{k_1}-x_1^{i_1} \right) + h_{21}e^{j\theta} \left( x_2^{k_2}-x_2^{i_2} \right)|^2}{2\sigma_1^2} \right)\right),~~~~~~~~~~~~~~~~~~~~~~~~~~~~~~~~~~~~~~~~~~~~~~~~~~~~~~~~~~~~~~~~~~~~~~~~~~~~~~~~~\right\}\\
\label {eqnset6_7}
&\hspace{-15 cm} \left\{\hspace{20cm} \sum_{k_1 = 0}^{M_1-1} \sum_{k_2 = 0}^{M_2-1} log \left( \sum_{i_1=0}^{M_1-1} \sum_{i_2=0}^{M_2-1} exp \left(- \frac{|h_{12} \left( x_1^{k_1}-x_1^{i_1} \right) + e^{j\theta} \left( x_2^{k_2}-x_2^{i_2} \right)|^2}{2\sigma_2^2} \right)\right)~\right\}
\end{align}
\hrule
\end{figure*}
\begin{theorem}
 \label{thm1}
Given the constellation pair (${\cal S}_1$,${\cal S}_2$) for the users, an approximate angle of rotation $\theta_{opt}$ for ${\cal S}_2$ required to maximally enlarge the CC capacity region of the GSIC at high power levels is given by (\ref{eqnset6_7}).
\end{theorem}
\begin{proof}
Define $I_1,~I_2,~I_1^ \prime,~I_2^ \prime$ as in (\ref{eqnset6_3})-(\ref{eqnset6_6}). Equations (\ref{eqnset6_4}) and (\ref{eqnset6_6}) follow from application of Jensen's Inequality on the expectation terms of $I_1$ and $I_2$ respectively. The required angle of rotation is:  $\theta_{opt}^\prime=arg~\max_{\theta \in (0,2\pi)}~min \lbrace I_1,I_2 \rbrace$.  Since closed form expressions for $I_1~\mbox{and}~I_2$ are not available, we maximize the minimum of the lower bounds on $I_1~\mbox{and}~I_2$, i.e. $\max_{\theta \in (0,2\pi)}~min \lbrace I_1^ \prime,I_2^ \prime \rbrace$. Canceling the common terms in $I_1^ \prime$ and $I_2^ \prime$ we arrive at the expression for $\theta_{opt}$ in (\ref{eqnset6_7}). At high power levels $P_1,P_2$, the CC capacity region obtained from $\theta_{opt}^\prime$ will be close to that obtained from $\theta_{opt}$. The proof for this is as follows:

Let $N_{1R}=Re(N_1),$ $N_{1I}=Im(N_1)$ and $p_{N_1}(n_1)$ be the pdf of the noise $N_1$ at $n_1$. Also, define
\begin{align}
\nonumber
\mu_1(k_1,k_2,i_1,i_2)=\left(x_1^{k_1}-x_1^{i_1} \right) + h_{21}e^{j\theta} \left( x_2^{k_2}-x_2^{i_2} \right)
\end{align} where $k_1 \mbox{ and } i_1$ can take values from $0$ to $(M_1-1)$, and $k_2 \mbox{ and } i_2$ can take values from $0$ to $(M_2-1)$.
We shall denote $\mu_1(k_1,k_2,i_1,i_2)$ as $\mu_1$ for short; for a given $\theta$, it is understood that $\mu_1$ is a function of $k_1,k_2,i_1, \mbox{ and } i_2$. Note that, for a given $\theta$ and $(k_1,k_2)$, and for $(i_1,i_2)$ $\neq$ $(k_1,k_2)$, the absolute value of $\mu_1$ gives the distance between two points in ${\cal S}_{sum_1}$. Now, for a fixed $(k_1,k_2)$ and $\theta$, define the set
\begin{align}
 M_1(k_1,k_2)=\left\{ (i_1,i_2)\neq (k_1,k_2) \mid \mu_1=0 \right\}
\end{align} $M_1(k_1,k_2)$ is the null set for all $(k_1,k_2)$ if the mapping ${\varphi}_1$ is one-one, else it is a non-empty set for some ($k_1,k_2$). Let $P_1^{k_1,k_2}=|M_1(k_1,k_2)|$. Now, consider the expression for $I_1$ in (\ref{eqnset6_3}). The expectation term in it is the only term dependent on $\theta$. So, consider $I_1^{\prime \prime}$, defined as in (\ref{metricpf0}), alternatively written as in (\ref{metricpf1}). The probability of the event $\left \{|N_{1R}|>\sqrt {2\sigma_1^2},~\mbox{and }|N_{1I}|> \sqrt {2\sigma_1^2}\right \}$ to occur is very small, as the variances of $N_{1R} \mbox{ and } N_{1I}$ are both equal to $\frac{\sigma_1^2}{2}$. Hence, the second integral in (\ref{metricpf1}) can be neglected. At high power levels, for a given ($k_1,k_2$) and a given $\theta$, (\ref{metricpf2}) is satisfied. The expression for $I_1^{\prime \prime}$ is further reduced to (\ref {metricpf5}), where, (\ref {metricpf3}) and (\ref {metricpf4}) follow from (\ref{metricpf2}) and the fact that $|n_{1R}|\leq \sqrt {2\sigma_1^2},|n_{1I}|\leq \sqrt {2\sigma_1^2}$, and the constant $c_1$ in (\ref{metricpf4.2}) arises from evaluation of the integral in (\ref{metricpf4.1}). 
\begin{figure*}
\scriptsize
\begin{align}
\label{metricpf2}
\min_{i_1,i_2,k_1,k_2} |\mu_1|>>2\sqrt {2\sigma_1^2}~~ \mbox{   where $(i_1,i_2)\neq(k_1,k_2)$ and $(i_1,i_2)\notin M_1(k_1,k_2)$ ; $0 \leq k_1,i_1 \leq (M_1-1)$ and $0 \leq k_2,i_2 \leq (M_2-1)$}
\end{align}
\hrule
\end{figure*}
\begin{figure*}
\scriptsize
\begin{align}
\label{metricpf0}
\mbox{Let }I_1^{\prime \prime}& \triangleq \sum_{k_1 = 0}^{M_1-1} \sum_{k_2 = 0}^{M_2-1}E_{N_1} \left[log \left( \sum_{i_1=0}^{M_1-1} \sum_{i_2=0}^{M_2-1} exp {\left(- \frac{|N_1 + \left(x_1^{k_1}-x_1^{i_1}\right) + h_{21}e^{j\theta} \left(x_2^{k_2}-x_2^{i_2} \right)|^2-|N_1|^2}{\sigma_1^2} \right)}\right)\right]\\
\nonumber
&=\sum_{k_1 = 0}^{M_1-1} \sum_{k_2 = 0}^{M_2-1}\left[\int_{|n_{1R}|\leq \sqrt {2\sigma_1^2},|n_{1I}|\leq \sqrt {2\sigma_1^2}} p_{N_1}(n_1)~log\left(1+ P_1^{k_1,k_2}+\sum_{(i_1,i_2)\neq(k_1,k_2);(i_1,i_2)\notin M_1(k_1,k_2)}e^{-\frac{|n_1+ \mu_1 |^2}{\sigma_1^2}}e^{\frac{|n_1|^2}{\sigma_1^2}}\right)\,dn_1 \hspace{30cm}\right]\\
\label{metricpf1}
&\hspace{-30cm} \left[ \hspace{33cm} +\int_{|n_{1R}|> \sqrt {2\sigma_1^2},|n_{1I}|> \sqrt {2\sigma_1^2}} p_{N_1}(n_1)~log\left(1+ P_1^{k_1,k_2}+\sum_{(i_1,i_2)\neq(k_1,k_2);(i_1,i_2)\notin M_1(k_1,k_2)}e^{-\frac{|n_1+ \mu_1 |^2}{\sigma_1^2}}e^{\frac{|n_1|^2}{\sigma_1^2}}\right)\,dn_1 \right]\\
&\approx \sum_{k_1 = 0}^{M_1-1} \sum_{k_2 = 0}^{M_2-1}\left[\int_{|n_{1R}|\leq \sqrt {2\sigma_1^2},|n_{1I}|\leq \sqrt {2\sigma_1^2}} p_{N_1}(n_1)~log\left(1+ P_1^{k_1,k_2}+\sum_{(i_1,i_2)\neq(k_1,k_2);(i_1,i_2)\notin M_1(k_1,k_2)}e^{-\frac{|n_1+ \mu_1 |^2}{\sigma_1^2}}e^{\frac{|n_1|^2}{\sigma_1^2}}\right)\,dn_1 \right]\\
\label{metricpf3}
&{\approx}~ \sum_{k_1 = 0}^{M_1-1} \sum_{k_2 = 0}^{M_2-1}\left[\int_{|n_{1R}|\leq \sqrt {2\sigma_1^2},|n_{1I}|\leq \sqrt {2\sigma_1^2}} p_{N_1}(n_1)~log\left(1+ P_1^{k_1,k_2}+\sum_{(i_1,i_2)\neq(k_1,k_2);(i_1,i_2)\notin M_1(k_1,k_2)}e^{-\frac{|\mu_1 |^2}{\sigma_1^2}}e^{\frac{|n_1|^2}{\sigma_1^2}}\right)\,dn_1 \right]\\
\label{metricpf4}
&{\approx}~ \sum_{k_1 = 0}^{M_1-1} \sum_{k_2 = 0}^{M_2-1}\left[\int_{|n_{1R}|\leq \sqrt {2\sigma_1^2},|n_{1I}|\leq \sqrt {2\sigma_1^2}} p_{N_1}(n_1)~log\left(1+ P_1^{k_1,k_2}+\sum_{(i_1,i_2)\neq(k_1,k_2);(i_1,i_2)\notin M_1(k_1,k_2)}e^{-\frac{|\mu_1 |^2}{\sigma_1^2}}\right)\,dn_1 \right]\\
\label{metricpf4.1}
&=\sum_{k_1 = 0}^{M_1-1} \sum_{k_2 = 0}^{M_2-1}log\left(1+ P_1^{k_1,k_2}+\sum_{(i_1,i_2)\neq(k_1,k_2);(i_1,i_2)\notin M_1(k_1,k_2)}e^{-\frac{|\mu_1 |^2}{\sigma_1^2}}\right)  \left[\int_{|n_{1R}|\leq \sqrt {2\sigma_1^2},|n_{1I}|\leq \sqrt {2\sigma_1^2}} p_{N_1}(n_1)\,dn_1 \right]\\
\label{metricpf4.2}
&=\sum_{k_1 = 0}^{M_1-1} \sum_{k_2 = 0}^{M_2-1} c_1 log\left(1+ P_1^{k_1,k_2}+\sum_{(i_1,i_2)\neq(k_1,k_2);(i_1,i_2)\notin M_1(k_1,k_2)}e^{-\frac{|\mu_1 |^2}{\sigma_1^2}}\right)\\
\label{metricpf5}
&=\sum_{k_1 = 0}^{M_1-1} \sum_{k_2 = 0}^{M_2-1} c_1  log\left(\sum_{i_1 = 0}^{M_1-1} \sum_{i_2 = 0}^{M_2-1}e^{-\frac{|\mu_1 |^2}{\sigma_1^2}}\right) 
\end{align}
\hrule
\end{figure*}
We now carry out the same procedure for $I_2$ also. Define
\begin{align}
\nonumber
\mu_2(k_1,k_2,i_1,i_2)=h_{12}\left(x_1^{k_1}-x_1^{i_1} \right) + e^{j\theta} \left( x_2^{k_2}-x_2^{i_2} \right)
\end{align} where $k_1 \mbox{ and } i_1$ can take values from $0$ to $(M_1-1)$, and $k_2 \mbox{ and } i_2$ can take values from $0$ to $(M_2-1)$.
We shall denote $\mu_2(k_1,k_2,i_1,i_2)$ as $\mu_2$ for short. Now, for a fixed ($k_1,k_2$) and $\theta$, define the set
\begin{align}
 M_2(k_1,k_2)=\left\{ (i_1,i_2)\neq (k_1,k_2) \mid \mu_2=0 \right\}
\end{align} $M_2(k_1,k_2)$ is the null set for all ($k_1,k_2$) if the mapping ${\varphi}_2$ is one-one, else it is a non-empty set for some ($k_1,k_2$). Since, the expectation term of $I_2$ in (\ref{eqnset6_5}) is the only term dependent on $\theta$, consider, $I_2^{\prime \prime}$ defined as in (\ref{metricpf0.1}). At high power levels, for a given ($k_1,k_2$) and a given $\theta$, (\ref{metricpf2.1}) is satisfied. Following similar steps as for $I_1^{\prime \prime}$, expression for $I_2^{\prime \prime}$ reduces to (\ref{metricpf5.1}). Since, $\sigma_1^2=\sigma_2^2$, $N_1$ and $N_2$ have the same distribution and hence $c_1=c_2$.
\begin{figure*}
\scriptsize
\begin{align}
\label{metricpf2.1}
\min_{i_1,i_2,k_1,k_2} |\mu_2|>>2\sqrt {2\sigma_2^2}~~ \mbox{   where $(i_1,i_2)\neq(k_1,k_2)$ and $(i_1,i_2)\notin M_2(k_1,k_2)$ ; $0 \leq k_1,i_1 \leq (M_1-1)$ and $0 \leq k_2,i_2 \leq (M_2-1)$}
\end{align}
\hrule
\end{figure*}
\begin{figure*}
\scriptsize
\begin{align}
\label{metricpf0.1}
\mbox{Let }I_2^{\prime \prime}& \triangleq \sum_{k_1 = 0}^{M_1-1} \sum_{k_2 = 0}^{M_2-1}E_{N_2} \left[log \left( \sum_{i_1=0}^{M_1-1} \sum_{i_2=0}^{M_2-1} exp {\left(- \frac{|N_2 +  h_{12}\left(x_1^{k_1}-x_1^{i_1}\right) + e^{j\theta} \left(x_2^{k_2}-x_2^{i_2} \right)|^2-|N_2|^2}{\sigma_2^2} \right)}\right)\right]\\
\label{metricpf5.1}
&\approx \sum_{k_1 = 0}^{M_1-1} \sum_{k_2 = 0}^{M_2-1} c_2 log\left(\sum_{i_1 = 0}^{M_1-1} \sum_{i_2 = 0}^{M_2-1}e^{-\frac{|\mu_2 |^2}{\sigma_2^2}}\right)  ~~~~~\mbox{ where, }c_2=\left[\int_{|n_{2R}|\leq \sqrt {2\sigma_2^2},|n_{2I}|\leq \sqrt {2\sigma_2^2}} p_{N_2}(n_2)\,dn_2 \right]
\end{align}
\hrule
\end{figure*}
Now, consider the terms in the metric for $\theta_{opt}$ in (\ref{eqnset6_7}), rewritten in terms of $\mu_1$ and $\mu_2$ in (\ref{metricpf6}) and (\ref{metricpf6.1}) respectively.
\begin{figure*}
\scriptsize
\begin{align}
\label{metricpf6}
&\sum_{k_1 = 0}^{M_1-1} \sum_{k_2 = 0}^{M_2-1} log \left( \sum_{i_1=0}^{M_1-1} \sum_{i_2=0}^{M_2-1} e^{ \left(- \frac{|\left( \left(x_1^{k_1}-x_1^{i_1}\right) \right) + h_{21}e^{j\theta} \left( x_2^{k_2}-x_2^{i_2} \right)|^2}{2\sigma_1^2} \right)}\right)=\sum_{k_1 = 0}^{M_1-1} \sum_{k_2 = 0}^{M_2-1} log \left( \sum_{i_1=0}^{M_1-1} \sum_{i_2=0}^{M_2-1} e^ {\left(- \frac{|\mu_1|^2}{2\sigma_1^2}\right)}\right)\\
\label{metricpf6.1}
&\sum_{k_1 = 0}^{M_1-1} \sum_{k_2 = 0}^{M_2-1} log \left( \sum_{i_1=0}^{M_1-1} \sum_{i_2=0}^{M_2-1} e^{ \left(- \frac{|\left( h_{12}\left(x_1^{k_1}-x_1^{i_1}\right) \right) + e^{j\theta} \left( x_2^{k_2}-x_2^{i_2} \right)|^2}{2\sigma_2^2} \right)}\right)=\sum_{k_1 = 0}^{M_1-1} \sum_{k_2 = 0}^{M_2-1} log \left( \sum_{i_1=0}^{M_1-1} \sum_{i_2=0}^{M_2-1} e^ {\left(- \frac{|\mu_2|^2}{2\sigma_2^2}\right)}\right)
\end{align}
\hrule
\end{figure*}
At high values values of $x$, the difference between $e^{-x^2}$ and $e^{-x^2/2}$ is very small. Hence, the expressions, (\ref{metricpf5}) divided by $c_1$ and (\ref{metricpf6}), and, (\ref{metricpf5.1}) divided by $c_2$ and (\ref{metricpf6.1}) give almost the same value at high powers. In other words, 
\begin{align}
\nonumber
\theta_{opt}^\prime&=\max_{\theta \in (0,2\pi)} min\{I_1,I_2\}=\min_{\theta \in (0,2\pi)} max\left\{\frac{I_1^{\prime\prime}}{c_1},\frac{I_2^{\prime\prime}}{c_1}\right \}\approx \theta_{opt}.
\end{align}
\end{proof}

\begin{table*}
\caption{Optimum angle of rotation and sum-capacities for QPSK alphabet pair (${\cal S}_1,{\cal S}_2$) for some values of channel gains and Powers.}
\centering
\scriptsize
\begin{tabular}{|c|c|c|c|c|c|c|c|c|}
\hline
$P_1$ & $P_2$ & $h_{12}$ & $h_{21}$ & $\theta_{opt}$ & $\theta_{opt}^\prime$ & Max. CC Sum & Max. CC Sum & Max. CC Sum \\
(Watt) & (Watt) &  &  &  &  &  Capacity (Unrotated)  &  Capacity (Rotated by $\theta_{opt}$) & Capacity (Rotated by $\theta_{opt}^\prime$)\\
\hline
$3.5$ & $6$ & $1 \angle 10^\circ$ & $1 \angle 20^\circ$ & $39.53^\circ$ & $41.25^\circ$ & $3.006$ & $3.107$ & $3.108$ \\
\hline
$3.5$ & $6$ & $1.2 \angle 10^\circ$ & $1.1 \angle 20^\circ$ & $46.41^\circ$ & $44.69^\circ$ & $2.994$ & $3.22$ & $3.221$ \\
\hline
$5$ & $5$ & $1.2 \angle 15^\circ$ & $1.5 \angle 5^\circ$ & $73.91^\circ$ & $72.19^\circ$ & $3.178$ & $3.319$ & $3.32$ \\
\hline
$8$ & $6$ & $1.8 \angle 40^\circ$ & $1.3 \angle 70^\circ$ & $49.85^\circ$ & $51.57^\circ$ & $3.459$ & $3.577$ & $3.58$ \\
\hline
\end{tabular}
\hrule
\label{table1}
\end{table*}
Note that the metric is easy to evaluate as it does not involve $N_1$ and $N_2$. On the contrary, $\theta_{opt}^\prime$ has to be evaluated numerically, as done in \cite{KnS}. The metric works well, as illustrated by Fig. \ref{fig:arbitchgains_qpsk} and some simulation results in Table \ref{table1}, (the channel gains and powers are chosen randomly,)  where the capacity regions obtained from $\theta_{opt}$ and $\theta_{opt}^\prime$ are too close to each other\footnote{In all the plots, in this paper, ``Rotated Acc. to Numerically Computed angle'' refers to the CC capacity according to rotation by $\theta_{opt}^\prime$ which is computed numerically, ``Rotated Acc. to Metric'' refers to the CC capacity according to rotation by $\theta_{opt}$ and ``Sum'' refers to the maximum sum rate $R_1+R_2$ on the respective curves.} (the last two columns). For a given constellation pair, there will be a significant change in the CC capacity due to rotation only at high powers. The reason for this is given below.

\begin{figure}[htbp]
\centering
\includegraphics[totalheight=2.5in,width=3.5in]{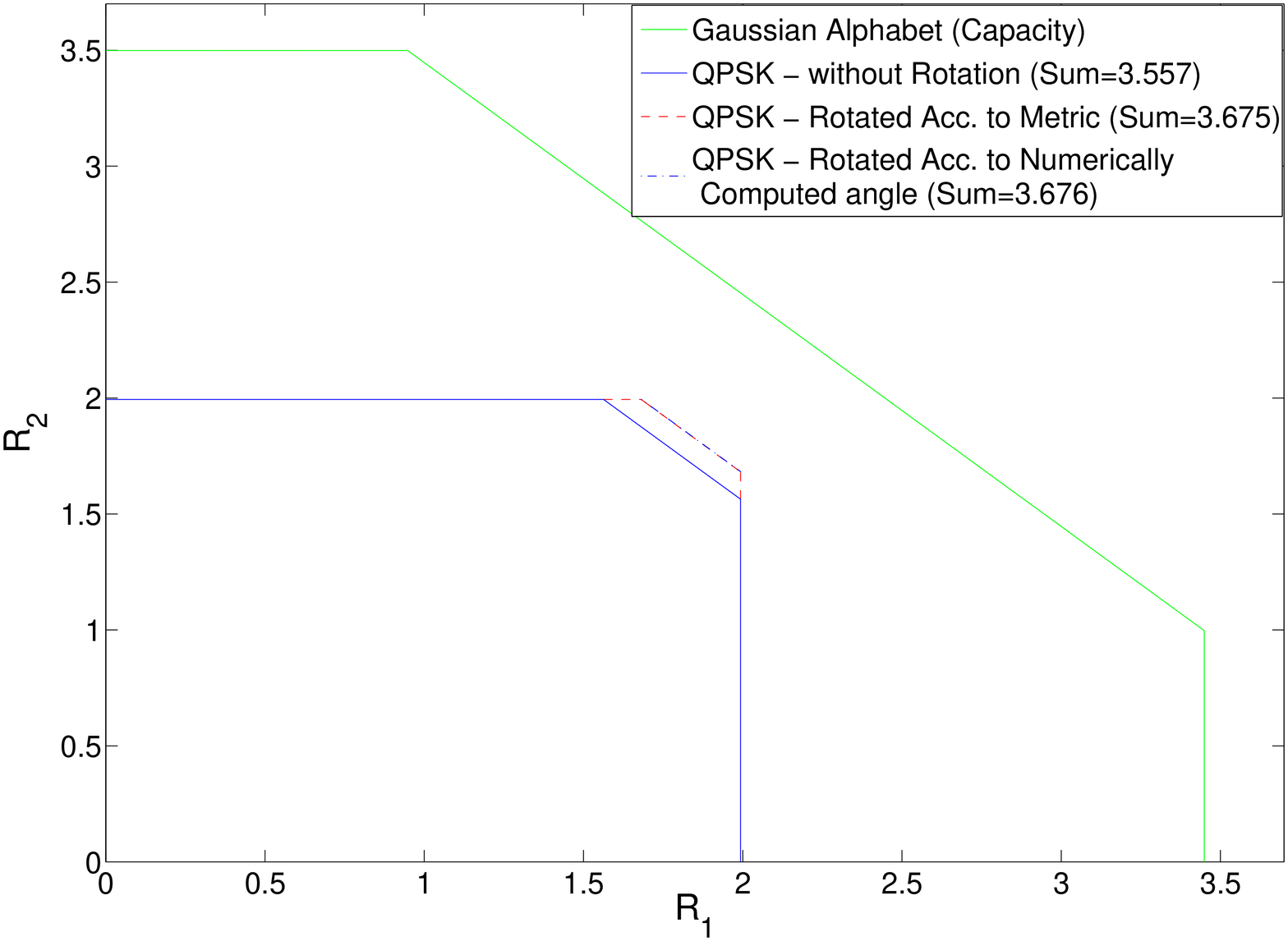}
\caption{CC capacity for QPSK pair (${\cal S}_1,{\cal S}_2$) with $P_1$=$9.92$ Watt (=$9.96dB$), $P_2$=$10.3$ Watt (=$10.13dB$), $n_1\mbox{$=$}n_2\mbox{$=$}1$, $h_{12}$=$1.03\angle-112^\circ$, $h_{21}$=$1.07\angle-44^\circ$, $\theta_{opt}^\prime$=$79.0682^\circ$, $\theta_{opt}$=$ 77.3493^\circ$. The curves corresponding to $\theta_{opt}^\prime$ and $\theta_{opt}$ are close and hence, indistinguishable.}	
\label{fig:arbitchgains_qpsk}	
\end{figure}

The sphere packing argument for the G-MAC, in \cite{HaR}, which explained why the capacity does not improve much with rotation at low $SNR$ can be extended to the general GSIC as follows: Fixed powers ($P_1,P_2$) and channel gains $h_{ij}$ ($i,j=1,2$), which can take complex values, can correspond to fixed radius, $r_1$ and $r_2$, of two dimensional balls, $B_{r_1}$ and $B_{r_2}$ respectively, and the signal points in the sum-constellation ${\cal S}_{sum_i}$ can correspond to points inside its ball $B_{r_i}$ ($i=1,2$). As the number of points in at least one of the input constellations (${\cal S}_1,{\cal S}_2$) increases, the number of points, $M_i=\mid{\cal S}_{sum_i}\mid$ in $B_{r_i}$ increases and hence the density of points in $B_{r_i}$ ($i=1,2$) increases. From (\ref{eqnset6_3}) and (\ref{eqnset6_5}), it can be seen that the CC capacity  depends on the distance distribution of the points of ${\cal S}_{sum_i}$ in $B_{r_i}$ ($i=1,2$). It is clear that rotation of one of the constellations, will cause perturbations in ${\cal S}_{sum_i}$, and hence its points in $B_{r_i}$ ($i=1,2$) gets rearranged. For large values of $M_1$ or $M_2$, even though the points in $B_{r_i}$ ($i=1,2$) rearrange themselves as a result of rotation, the density of $B_{r_i}$ is so large that the distance distribution of the points inside the balls change negligibly and as a result of (\ref{eqnset6_3}) and (\ref{eqnset6_5}), there is not much change in the CC capacity due to rotation. Fig. \ref{fig:arbitchgains_qpsk} and Fig. \ref{fig:arbitchgains_psk8} illustrate this argument. At the same values of channel gains and powers there is negligible improvement in the CC capacity for the 8-PSK pair (${\cal S}_1,{\cal S}_2$), shown in  Fig. \ref{fig:arbitchgains_psk8}, while there is good improvement in CC capacity for the
QPSK pair (${\cal S}_1,{\cal S}_2$) shown in  Fig. \ref{fig:arbitchgains_qpsk}.
\begin{figure}[htbp]
\centering
\includegraphics[totalheight=2.5in,width=3.5in]{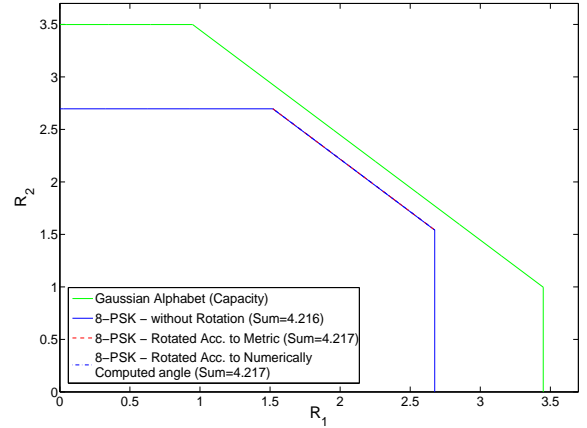}
\caption{CC capacity for 8-PSK pair (${\cal S}_1,{\cal S}_2$) with $P_1$=$9.92$ Watt (=$9.96dB$), $P_2$=$10.3$ Watt (=$10.13dB$), $n_1\mbox{$=$}n_2\mbox{$=$}1$, $h_{12}$=$1.03\angle-112^\circ$, $h_{21}$=$1.07\angle-44^\circ$.}
\label{fig:arbitchgains_psk8}	
\end{figure}
The arguments tally with the Gaussian input alphabet case where the input alphabets are unconstrained and the capacity remains invariant to rotation. So, at fixed channel gains, for rotation to have considerable effect on the CC capacity of a finite constellation pair, the powers should be commensurate with the size of the constellation, and hence the powers should be high enough. So, we need to rotate the constellation only at sufficiently high powers.
%

\section{SUBOPTIMALITY OF FDMA WITH FINITE CONSTELLATIONS}
\label{sec4}
FDMA with finite input constellation for two-user GMAC was first plotted in \cite{HaR2} and some interesting comparisons with behaviour for Gaussian alphabets were made. In the two-user GSIC, it is known that, for the Gaussian alphabet case with $P_1$=$P_2$, when $SNR_1$=$SNR_2$=$INR_1$=$INR_2$, the FDMA rate curve touches the capacity curve \cite{GaY}. It is also shown in \cite{GaY}, for the Gaussian alphabet case with $P_1$=$P_2$, that the FDMA inner-bound, at the optimum sum-rate point, is better than the simultaneous-decoding inner-bound in the weak interference regime when $INR_1$=$INR_2$$<$$SNR_1$=$SNR_2$. In this section, we show that, for the Gaussian alphabet case with $P_1$ not necessarily equal to $P_2$, the FDMA rate curve touches the capacity curve of the GSIC when $SNR_1$$\leq$$INR_2$ and $SNR_2$=$INR_1$ or when $SNR_1$=$INR_2$ and $SNR_2$$\leq$$INR_1$ and throughout the weak-interference regime the FDMA inner-bound, at the optimum sum-rate point, is always better than  the simultaneous-decoding inner-bound. On the contrary, for the constellation constrained case, with $P_1$ not necessarily equal to $P_2$ and ${\cal S}_1$=${\cal S}_2$, we show that the FDMA rate curve does not touch the CC capacity curve throughout the strong-interference regime. We also show that, for a portion of the weak interference regime, under constellation constraints, the simultaneous-decoding inner-bound, enlarged by relative rotation between the finite constellations, is strictly better than the FDMA inner-bound. Throughout the section we assume ${\cal S}_1$=${\cal S}_2$.

Since FDMA involves bandwidth we need to consider a modified channel model as described below.
\subsection{Model for CC Capacity with Full Bandwidth Usage}
The model of the two-user GIC (shown in Fig. \ref{fig:AWGNIC}) under strong interference considered in this section is similar to the one presented in Section \ref{sec1}. We point out only the changes in the signal model with reference to the model in Section \ref{sec1}. It is assumed that User-1 and User-2 communicate to the destination at the same time and in the same frequency band of W Hertz. To take into consideration the bandwidth, the variance of the additive noise at both the receivers are given by $WN_0$. The signals received at the destinations are given by
\begin{align}
\nonumber
&Y_1 = \sqrt{P_1}X_1 + h_{21}\sqrt{P_2}X_2 + N_1\\
 \label{eqnset9}
&Y_2 = h_{12}\sqrt{P_1}X_1 + \sqrt{P_2}X_2 + N_2,
\end{align}
$\mbox{where, }X_1 \in {\cal S}_1, X_2 \in {\cal S}_2e^{j \theta}$ (finite constellations ${\cal S}_1$ and ${\cal S}_2$ are be of unit power), $N_1 \sim {\cal {CN}}(0,WN_0)\mbox{ and } N_2 \sim {\cal {CN}}(0,WN_0)$ ($N_0/2$ is the power spectral density of the AWGN in each dimension). Without loss of generality we take $N_0=1$. We assume that every channel use consumes $T$ seconds for each user (where $\frac{1}{T}= W$ Hertz).

Applying the CC capacity regions used in Section \ref{sec2} to the channel model in (\ref{eqnset9}), the set of CC capacity values (in bits per channel use) that define the boundary of the CC capacity region, are given by
\begin{align}
\label{eqnset10a}
R_1& \leq I_W\left(\sqrt{P_1}X_1;Y_1|\sqrt{P_2}X_2\right)\\
\label{eqnset10b}
R_2& \leq I_W\left(\sqrt{P_2}X_2;Y_2|\sqrt{P_1}X_1\right)\\
\nonumber
R_1+R_2 & \leq min \left \{I_W\left(\sqrt{P_1}X_1,\sqrt{P_2}X_2;Y_1\right),\hspace{30cm} \right \}\\
\label{eqnset10}
&\hspace{-30cm} \left \{ \hspace{30cm}~~~~~~~~~~~~~~I_W\left(\sqrt{P_1}X_1,\sqrt{P_2}X_2;Y_2\right) \right \},
\end{align}
where, the expressions for the mutual informations in (\ref{eqnset10a}) and the first term of (\ref{eqnset10}) are given in (\ref{eqnset11a}) and (\ref{eqnset11b}) (shown at the top of next page) respectively, and, the expressions for the mutual informations in (\ref{eqnset10b}) and the second term of (\ref{eqnset10}) are similar to the ones in (\ref{eqnset11a}) and (\ref{eqnset11b}) respectively. We denote the mutual informations with subscript $W$ as they depend on the bandwidth $W$. The CC capacity is achieved by simultaneous-decoding scheme with finite input constellations.
\begin{figure*}
\scriptsize
\begin{align}
\centering
\label {eqnset11a}
&I_W(\sqrt{P_1}X_1;Y_1|\sqrt{P_2}X_2)= logM_1 - \frac{1}{M_1} \sum_{k_1 = 0}^{M_1-1} E_{N_1} \left[log \left( \sum_{i_1=0}^{M_1-1} exp \left(- \frac {\left( |N_1 + \sqrt{P_1} \left( x_1^{k_1}-x_1^{i_1} \right)|^2-|N_1|^2 \right)}{W} \right)\right)\right]\\
\centering
\nonumber
&I_W(\sqrt{P_1}X_1,\sqrt{P_2}X_2;Y_1)= log(M_1M_2)\\
\label {eqnset11b}
&~~~~~~~~~~~~~~~~~~~~~~~~~~~~~~~~~~~~ - \frac{1}{M_1M_2} \sum_{k_1 = 0}^{M_1-1} \sum_{k_2 = 0}^{M_2-1} E_{N_1} \left[log \left( \sum_{i_1=0}^{M_1-1} \sum_{i_2=0}^{M_2-1} exp \left(- \frac {\left( |N_1 + \sqrt{P_1} \left( x_1^{k_1}-x_1^{i_1} \right) + h_{21}e^{j\theta}\sqrt{P_2} \left( x_2^{k_2}-x_2^{i_2} \right)|^2-|N_1|^2 \right)}{W}\right)\right)\right]
\end{align}
\hrule
\end{figure*}

Since every channel use consumes $T$ seconds, the rate pairs (in bits per seconds) that define the CC capacity region  are given by
\begin{align}
\label{eqnset12a}
R_1& \leq WI_W\left(\sqrt{P_1}X_1;Y_1|\sqrt{P_2}X_2\right)\\
\label{eqnset12b}
R_2& \leq WI_W\left(\sqrt{P_2}X_2;Y_2|\sqrt{P_1}X_1\right)\\
\nonumber
R_1+R_2 & \leq min \left \{WI_W\left(\sqrt{P_1}X_1,\sqrt{P_2}X_2;Y_1\right),\hspace{30cm} \right \}\\
\label{eqnset12c}
&\hspace{-30cm} \left \{ \hspace{30cm}~~~~~~~~~~~~~~WI_W\left(\sqrt{P_1}X_1,\sqrt{P_2}X_2;Y_2\right) \right \}.
\end{align}
The capacity region of the strong interference channel is given by
\begin{align}
\label{eqnset13a}
R_1& \leq Wlog\left(1+\frac{P_1}{W}\right)\\
\label{eqnset13b}
R_2& \leq Wlog\left(1+\frac{P_2}{W}\right)\\
\nonumber
R_1+R_2 & \leq min \left \{Wlog\left(1+\frac{P_1+\mid h_{21}\mid^2P_2}{W}\right),\hspace{30cm} \right \}\\
\label{eqnset13c}
&\hspace{-30cm} \left \{ \hspace{30cm}~~~~~~~~~~~~~~Wlog\left(1+\frac{\mid h_{12}\mid^2P_1+P_2}{W}\right) \right \}.
\end{align}
The capacity can be achieved by simultaneous-decoding scheme, with Gaussian input alphabets.
\subsection{CC Capacity with FDMA}
User-1--Receiver-1 agree on $W_1=\alpha W$ bandwidth and User-2--Receiver-2 agree on the non-overlapping $W_2=(1-\alpha) W$ bandwidth, $0< \alpha < 1$. Hence, for each $i$ = $1, 2$, User-$i$, with bandwidth $W_i$ and power constraint $ P_i$, equipped with finite constellation $\sqrt{P_i} {\cal S}_i$, views a Single-Input Single-Output (SISO) AWGN channel with Receiver-$i$ without interference. The circularly symmetric Gaussian noise at the Receiver-$i$ has mean zero and variance $W_iN_0$ (and without loss of generality we assume $N_0=1$). Hence, the channel model is given by
\begin{align}
&Y_1 = \sqrt{P_1}X_1 + N_1\\
 \label{eqnsetFDMACC}
&Y_2 = \sqrt{P_2}X_2 + N_2,
\end{align}
$\mbox{where, }X_1\in {\cal S}_1, X_2\in {\cal S}_2$ (${\cal S}_1$ and ${\cal S}_2$ are taken to be of unit power), $N_i \sim {\cal {CN}}(0,W_i)\mbox{ ($i=1,2$)}$.

The maximum achievable rate pair (in bits per second) for the two users, under constellation constraints, are given by
\begin{align}
\label{eqnset14a}
&R_1 \leq W_1I_{W_1}\left(\sqrt{P_1}X_1;Y_1|\sqrt{P_2}X_2\right)\\
\label{eqnset14b}
&R_2 \leq W_2I_{W_2}\left(\sqrt{P_2}X_2;Y_2|\sqrt{P_1}X_1\right).
\end{align}
Therefore, the sum-rate region achievable  with FDMA, under constellation constraints, is given by
\begin{align}
\nonumber
&R_1+R_2 \leq W_1I_{W_1}\left(\sqrt{P_1}X_1;Y_1|\sqrt{P_2}X_2\right)\\
\label{eqnset15}
&~~~~~~~~~~~~~~~~~~~~~~+W_2I_{W_2}\left(\sqrt{P_2}X_2;Y_2|\sqrt{P_1}X_1\right).
\end{align}

With Gaussian input alphabets, the achievable rate pair for FDMA is given by 
\begin{align}
\label{eqnset16a}
&R_1 \leq W_1log\left(1+ \frac{P_1}{W_1}\right)\\
\label{eqnset16b}
&R_2 \leq W_2log\left(1+ \frac{P_2}{W_2}\right).
\end{align}
%

The following theorems show that, in the finite constellation case, the $\alpha$ that would maximize the sum rate for FDMA is the same as that in the Gaussian alphabet case.
\begin{theorem}
 \label{thm2}
For the GIC model in Fig. \ref{fig:AWGNIC}, when ${\cal S}_1={\cal S}_2$, the value of $\alpha$ that would maximize the sum rate for FDMA, in the finite constellation case, is equal to $\frac{P_1}{P_1+P_2}$.
\end{theorem}
\begin{proof}
\begin{figure*}
\scriptsize
\begin{align}
\label{thm2_pf1}
W_1 I_{W_1}(\sqrt{P_1}X_1;Y_1|\sqrt{P_2}X_2)&= \alpha W \left(logM_1 - log~ e - \frac{1}{M_1} \sum_{k_1 = 0}^{M_1-1} E_{N_1} \left[log \left( \sum_{i_1=0}^{M_1-1} exp \left(- \frac {\left( |N_1 +\sqrt{P_1} \left( x_1^{k_1}-x_1^{i_1} \right)|^2 \right)}{\alpha W} \right)\right)\right]\right)\\
\label{thm2_pf2}
W_2I_{W_2}(\sqrt{P_2}X_2;Y_2|\sqrt{P_1}X_1)&= (1-\alpha) W \left(logM_2 - log~ e - \frac{1}{M_2} \sum_{k_2 = 0}^{M_2-1} E_{N_2} \left[log \left( \sum_{i_2=0}^{M_2-1} exp \left(- \frac {\left( |N_2 + \sqrt{P_2} \left( x_2^{k_2}-x_2^{i_2} \right)|^2 \right)}{(1-\alpha) W} \right)\right)\right]\right)\\
\label{thm2_pf3}
{\cal I}_1& \triangleq \alpha W \sum_{k_1 = 0}^{M_1-1} E_{N_1} \left[log \left( \sum_{i_1=0}^{M_1-1} exp \left(- \frac {\left( |N_1 + \sqrt{P_1} \left( x_1^{k_1}-x_1^{i_1} \right)|^2 \right)}{\alpha W} \right)\right)\right]\\
\label{thm2_pf4}
{\cal I}_2& \triangleq (1-\alpha) W \sum_{k_2 = 0}^{M_2-1} E_{N_2} \left[log \left( \sum_{i_2=0}^{M_2-1} exp \left(- \frac {\left( |N_2 + \sqrt{P_2} \left( x_2^{k_2}-x_2^{i_2} \right)|^2 \right)}{(1-\alpha) W} \right)\right)\right]
\end{align}
\hrule
\end{figure*}
\begin{figure*}
\scriptsize
\begin{align}
\label{thm2_pf5}
&\frac{d}{d\alpha}\left(W_1 I_{W_1}(\sqrt{P_1}X_1;Y_1|\sqrt{P_2}X_2)+W_2I_{W_2}(\sqrt{P_2}X_2;Y_2|\sqrt{P_1}X_1) \right)=0\\
\label{thm2_pf6}
\Rightarrow~& WlogM_1-WlogM_2-Wlog~e+Wlog~e-\frac{1}{M_1}\frac{d {\cal I}_1}{d\alpha}-\frac{1}{M_2}\frac{d {\cal I}_2}{d\alpha}=0\\
\label{thm2_pf7}
\Rightarrow~&\frac{d {\cal I}_1}{d\alpha}+\frac{d {\cal I}_2}{d\alpha}=0
\end{align}
\hrule
\end{figure*}
\begin{figure*}
\scriptsize
\begin{align}
\label{thm2_pf8}
&\frac{d {\cal I}_1}{d\alpha}=\left(\sum_{k_1 = 0}^{M_1-1} \frac{W}{\sqrt{\pi \alpha W}} \int_{0<|n_1|<\infty} e^{- \frac{|n_1|^2}{\alpha W}}\left[\frac{|n_1|^2}{\alpha W}~log\left(\sum_{i_1=0}^{M_1-1} e^{-\frac{|n_1+\mu_1|^2}{\alpha W}}\right)+\frac{\left(\sum_{i_1=0}^{M_1-1}e^{-\frac{|n_1+\mu_1|^2}{\alpha W}}~\frac{|n_1+\mu_1|^2}{\alpha W}\right)}{\sum_{i_1=0}^{M_1-1} e^{-\frac{|n_1+\mu_1|^2}{\alpha W}}}\right]\,dn_1\right)~+\frac{1}{2}{\cal I}_1^\prime.\\
\label{thm2_pf9}
&\frac{d {\cal I}_2}{d\alpha}=\left(\sum_{k_2 = 0}^{M_2-1} \frac{W}{\sqrt{\pi (1-\alpha) W}} \int_{0<|n_2|<\infty} e^{- \frac{|n_2|^2}{(1-\alpha) W}}\left[-\frac{|n_2|^2}{(1-\alpha) W}~log\left(\sum_{i_2=0}^{M_2-1} e^{-\frac{|n_2+\mu_2|^2}{(1-\alpha) W}}\right)-\frac{\left(\sum_{i_2=0}^{M_2-1}e^{-\frac{|n_2+\mu_2|^2}{(1-\alpha) W}}~\frac{|n_2+\mu_2|^2}{(1-\alpha) W}\right)}{\sum_{i_2=0}^{M_2-1} e^{-\frac{|n_2+\mu_2|^2}{(1-\alpha) W}}}\right]\,dn_2\right)~-\frac{1}{2}{\cal I}_2^\prime.
\end{align}
\hrule
\end{figure*}
The expressions for the maximum achievable rates with FDMA, under constellation constraints, in (\ref{eqnset14a}) and (\ref{eqnset14b}) are given in (\ref{thm2_pf1}) and (\ref{thm2_pf2}). Define ${\cal I}_1$ and ${\cal I}_2$ as in (\ref{thm2_pf3}) and (\ref{thm2_pf4}). It is required to find
\begin{align}
\nonumber
&\alpha_{opt}= arg \max_{\alpha \in (0,1)} \left(W_1 I_{W_1}(\sqrt{P_1}X_1;Y_1|\sqrt{P_2}X_2) \hspace{30cm}\right)\\
\nonumber
&\hspace{-30cm}\left(\hspace{33cm}+W_2I_{W_2}(\sqrt{P_2}X_2;Y_2|\sqrt{P_1}X_1) \right).
\end{align}
Therefore, at $\alpha=\alpha_{opt}$, (\ref{thm2_pf5})-(\ref{thm2_pf7}) are satisfied. As ${\cal S}_1={\cal S}_2$, (\ref{thm2_pf6}) reduces to (\ref{thm2_pf7}). Now, define ${\cal I}_1^\prime={\cal I}_1/ \alpha$ and ${\cal I}_2^\prime={\cal I}_2/ (1-\alpha)$. Also, define $\mu_1(k_1,i_1)=\sqrt{P_1} \left( x_1^{k_1}-x_1^{i_1} \right)$ and $\mu_2(k_2,i_2)=\sqrt{P_2} \left( x_2^{k_2}-x_2^{i_2} \right)$. We denote $\mu_1(k_1,i_1)$ as simply $\mu_1$ and $\mu_2(k_2,i_2)$ as $\mu_2$; it is understood that $\mu_1$ and $\mu_2$ are functions of $(k_1,i_1)$ and $(k_2,i_2)$ respectively. Expressions for $\frac{d {\cal I}_1}{d\alpha}$ and $\frac{d {\cal I}_2}{d\alpha}$ are given in (\ref{thm2_pf8}) and (\ref{thm2_pf9}), where, in (\ref{thm2_pf8}) and (\ref{thm2_pf9}), $n_1$ and $n_2$ are realizations of $N_1$ and $N_2$ respectively. Let $\alpha^{\prime\prime}= \frac{P_1}{P_1+P_2}$. Now, substitute $n_1^\prime=n_1/\sqrt{\alpha}$ in (\ref{thm2_pf8}) and $n_2^\prime=n_2/\sqrt{(1-\alpha)}$ in (\ref{thm2_pf9}). After this substitution, it can be easily seen that, at $\alpha=\alpha^{\prime\prime}= \frac{P_1}{P_1+P_2}$, ${\cal I}_1^\prime={\cal I}_2^\prime$ and the first term in in (\ref{thm2_pf8}) and the first term in (\ref{thm2_pf9}) are equal but for the sign. Hence, at $\alpha=\alpha^{\prime\prime}$, (\ref{thm2_pf5})-(\ref{thm2_pf7}) are satisfied. To, prove that $\alpha^{\prime\prime}$ = $\alpha_{opt}$, we need to show that the sum-rate $R_1+R_2$, achievable with FDMA, is a concave function of $\alpha \in (0,1)$, for which, it is enough to show that there exists a point on the FDMA rate curve in the ($R_1,R_2$) plane which achieves a greater sum rate than is achieved at a point on the line joining any two points on the curve. At this point where the sum rate is greater the sum rate achieved at a point on the line joining any two given points on the curve, the value of $\alpha$ must lie between the values of $\alpha$ at the given points. Let the points $A$ and $B$ lie on the FDMA curve in the ($R_1,R_2$) plane and let their co-ordinates be ($R_1^1,R_2^1$) and ($R_1^2,R_2^2$) respectively. Also, let the bandwidth-sharing parameter, $\alpha$, at the points $A$ and $B$ be $\alpha_1$ and $\alpha_2$ ($0<\alpha_1,\alpha_2<1$) respectively. The points ($R_1^1,R_2^1$) and ($R_1^2,R_2^2$) are defined by their respective expressions similar to the ones in (\ref{thm2_pf1}) and (\ref{thm2_pf2}). Let $W_1^1=\alpha_1 W$, $W_2^1=(1-\alpha_1) W$, $W_1^2=\alpha_2 W$, $W_2^2=(1-\alpha_2) W$ and also, define $f_1(\frac{1}{\alpha_i})$ and $f_2(\frac{1}{1-\alpha_i})$ ($i$=$1,2$) as in (\ref{thm2_pf10}) and (\ref{thm2_pf11}) respectively.
\begin{figure*}
\scriptsize
\begin{align}
\label{thm2_pf10}
&\frac{R_1^i}{\alpha_iW}=I_{{W_1}^i}\left(\sqrt{P_1}X_1;Y_1|\sqrt{P_2}X_2\right) \triangleq f_1\left(\frac{1}{\alpha_i}\right)\\ 
\label{thm2_pf11}
&\frac{R_2^i}{(1-\alpha_i)W}=I_{{W_2}^i}\left(\sqrt{P_2}X_2;Y_2|\sqrt{P_1}X_1\right) \triangleq f_2\left(\frac{1}{1-\alpha_i}\right)
\end{align}
\hrule
\end{figure*}
\begin{figure*}
\scriptsize
\begin{align}
\label{thm2_pf12}
R_1^{\prime\prime}&=\beta R_{1}^1+(1-\beta)R_{1}^2 =\alpha^\prime\left[\frac{\beta}{\alpha^\prime} R_{1}^1+\frac{(1-\beta)}{\alpha^\prime}R_{1}^2\right]\\
\label{thm2_pf13}
&=W\alpha^\prime\left[\frac{\beta}{\alpha^\prime} \alpha_1 f_1\left(\frac{1}{\alpha_1}\right) + \frac{(1-\beta)}{\alpha^\prime} \alpha_2 f_1\left(\frac{1}{\alpha_2}\right)\right]\\
\label{thm2_pf14}
& < W \alpha^{\prime} \left[ f_1\left(\frac{\beta+(1-\beta)}{\alpha^\prime}\right)\right] = W\alpha^\prime I_{{W_1}^\prime}\left(\sqrt{P_1}X_1;Y_1|\sqrt{P_2}X_2\right)\\
\label{thm2_pf15}
R_2^{\prime\prime}&=\beta R_{2}^1+(1-\beta)R_{2}^2<W(1-\alpha^\prime) I_{{W_2}^\prime}\left(\sqrt{P_2}X_2;Y_2|\sqrt{P_1}X_1\right)
\end{align}
\hrule
\end{figure*}
To achieve a point on the line joining the points $A$ and $B$, we need to time-share between the points $A$ and $B$, for a fraction of time $\beta$ and ($1-\beta$) ($0<\beta<1$) respectively. Now, let, $\beta\alpha_1+(1-\beta)\alpha_2 = \alpha^\prime$, $\beta(1-\alpha_1)+(1-\beta)(1-\alpha_2)=(1-\alpha^\prime)$, ${W_1}^\prime=\alpha^\prime W$, and ${W_2}^\prime=(1-\alpha^\prime) W$. The rate-pair, ($R_1^{\prime\prime},R_2^{\prime\prime}$), achieved by time-sharing between the points $A$ and $B$ is given in (\ref{thm2_pf12}) and (\ref{thm2_pf15}). Equation (\ref{thm2_pf14}) follows from the fact that $f_1$ is a concave function of $1/\alpha$ and, so, we apply Jensen's inequality in (\ref{thm2_pf13}) to arrive at (\ref{thm2_pf14}). Similarly, we arrive at (\ref{thm2_pf15}). Equations (\ref{thm2_pf14}) and (\ref{thm2_pf15}) imply that there exists a point on the FDMA curve in the ($R_1,R_2$) plane which achieves a greater sum rate than is achieved on the line joining the two points ($A$,$B$) on the curve and $\alpha^\prime$ lies between $\alpha_1$ and $\alpha_2$.
Hence, $\alpha^{\prime\prime}=\frac{P_1}{P_1+P_2}$ is the required optimum $\alpha$, i.e. $\alpha_{opt}$.
\end{proof}

\begin{theorem}
 \label{thm3}
For the GIC model in Fig. \ref{fig:AWGNIC}, the value of $\alpha$ that would maximize the sum rate for FDMA, in the Gaussian alphabet case, is equal to $\frac{P_1}{P_1+P_2}$.
\end{theorem}
\begin{proof}
The expressions for the maximum achievable rates with FDMA, in the Gaussian alphabet case, is given in  (\ref{eqnset16a}) and (\ref{eqnset16b}). Define $R_1^{c}$ and $R_2^{c}$ as given below.
\begin{align}
\label{thm3_pf1}
&R_1^{c} \triangleq \alpha W~ log\left(1+\frac{P_1}{\alpha W}\right)\\
\label{thm3_pf2}
&R_2^{c} \triangleq (1-\alpha) W~ log\left(1+\frac{P_2}{(1-\alpha) W}\right).
\end{align}
$R_1^{c}$ and $R_2^{c}$ define the points on the FDMA rate curve. It is required to find $\alpha_{opt}= arg~ \max_{\alpha \in (0,1)} \left(R_1^{c}+R_2^{c} \right)$. Therefore, at $\alpha=\alpha_{opt}$, (\ref{thm3_pf3}) (given at the top of the next page) is satisfied.
\begin{figure*}
\scriptsize
\begin{align}
\label{thm3_pf3}
\frac{d}{d \alpha} \left(R_1^{c}+R_2^{c}\right)=0 \Rightarrow Wlog\left(1+\frac{P_1}{\alpha W}\right) - \frac{\left(P_1/\alpha\right)}{\left(1+\frac{P_1}{\alpha W}\right)} - Wlog\left(1+\frac{P_2}{(1-\alpha) W}\right) + \frac{\left(P_2/(1-\alpha)\right)}{\left(1+\frac{P_2}{(1-\alpha) W}\right)} = 0
\end{align}
\hrule
\end{figure*}
Let $\alpha^{\prime\prime}= \frac{P_1}{P_1+P_2}$. It is easy to see that (\ref{thm3_pf3}) is satisfied at $\alpha=\alpha^{\prime\prime}$. To prove that $\alpha^{\prime\prime}$ = $\alpha_{opt}$, we need to show that the sum-rate $R_1^c+R_2^c$, achievable with FDMA, is a concave function of $\alpha \in (0,1)$, for which, it is enough to show that there exists a point on the FDMA rate curve in the ($R_1,R_2$) plane which achieves a greater sum rate than is achieved at a point on the line joining any two points on the curve. At this point where the sum rate is greater the sum rate achieved at a point on the line joining any two given points on the curve, the value of $\alpha$ must lie between the values of $\alpha$ at the given points. Let the points $A$ and $B$ lie on the FDMA curve in the ($R_1,R_2$) plane and let their co-ordinates be ($R_1^1,R_2^1$) and ($R_1^2,R_2^2$) respectively. Also, let the bandwidth-sharing parameter, $\alpha$, at the points $A$ and $B$ be $\alpha_1$ and $\alpha_2$ ($0<\alpha_1,\alpha_2<1$) respectively. The points ($R_1^1,R_2^1$) and ($R_1^2,R_2^2$) are defined by their respective expressions similar to the ones in (\ref{thm3_pf1}) and (\ref{thm3_pf2}).  Let $W_1^1=\alpha_1 W$, $W_2^1=(1-\alpha_1) W$, $W_1^2=\alpha_2 W$, $W_2^2=(1-\alpha_2) W$ and also, define $f_1(\frac{1}{\alpha_i})$ and $f_2(\frac{1}{1-\alpha_i})$ ($i$=$1,2$) as in (\ref{thm3_pf4}) and (\ref{thm3_pf5}) respectively.
\begin{figure*}
\scriptsize
\begin{align}
\label{thm3_pf4}
&\frac{R_1^i}{\alpha_iW}= log\left(1+\frac{P_1}{\alpha_i W}\right) \triangleq f_1\left(\frac{1}{\alpha_i}\right)\\ 
\label{thm3_pf5}
&\frac{R_2^i}{(1-\alpha_i)W}=log\left(1+\frac{P_2}{(1-\alpha_i) W}\right) \triangleq f_2\left(\frac{1}{1-\alpha_i}\right)
\end{align}
\hrule
\end{figure*}
\begin{figure*}
\scriptsize
\begin{align}
\label{thm3_pf6}
R_1^{\prime\prime}&=\beta R_{1}^1+(1-\beta)R_{1}^2 =\alpha^\prime\left[\frac{\beta}{\alpha^\prime} R_{1}^1+\frac{(1-\beta)}{\alpha^\prime}R_{1}^2\right]\\
\label{thm3_pf7}
&=W\alpha^\prime\left[\frac{\beta}{\alpha^\prime} \alpha_1 f_1\left(\frac{1}{\alpha_1}\right) + \frac{(1-\beta)}{\alpha^\prime} \alpha_2 f_1\left(\frac{1}{\alpha_2}\right)\right]\\
\label{thm3_pf8}
& < W \alpha^{\prime} \left[ f_1\left(\frac{\beta+(1-\beta)}{\alpha^\prime}\right)\right] = W\alpha^\prime log\left(1+\frac{P_1}{\alpha^\prime W}\right)\\
\label{thm3_pf9}
R_2^{\prime\prime}&=\beta R_{2}^1+(1-\beta)R_{2}^2<W(1-\alpha^\prime) log\left(1+\frac{P_2}{(1-\alpha^\prime) W}\right)
\end{align}
\hrule
\end{figure*}
To achieve a point on the line joining the points $A$ and $B$, we need to time-share between the points $A$ and $B$, for a fraction of time $\beta$ and ($1-\beta$) ($0<\beta<1$) respectively. Now, let, $\beta\alpha_1+(1-\beta)\alpha_2 = \alpha^\prime$, $\beta(1-\alpha_1)+(1-\beta)(1-\alpha_2)=(1-\alpha^\prime)$, ${W_1}^\prime=\alpha^\prime W$, and ${W_2}^\prime=(1-\alpha^\prime) W$. The rate-pair, ($R_1^{\prime\prime},R_2^{\prime\prime}$), achieved by time-sharing between the points $A$ and $B$ is given in (\ref{thm3_pf6}) and (\ref{thm3_pf9}). Equation (\ref{thm3_pf8}) follows from the fact that $f_1$ is a concave function of $1/\alpha$ and, so, we apply Jensen's inequality in (\ref{thm3_pf7}) to arrive at (\ref{thm3_pf8}). Similarly, we arrive at (\ref{thm2_pf9}). Equations (\ref{thm3_pf8}) and (\ref{thm3_pf9}) imply that there exists a point on the FDMA curve in the ($R_1,R_2$) plane which achieves a greater sum rate than is achieved on the line joining the two points ($A$,$B$) on the curve and $\alpha^\prime$ lies between $\alpha_1$ and $\alpha_2$.
Hence, $\alpha^{\prime\prime}=\frac{P_1}{P_1+P_2}$ is the required optimum $\alpha$, i.e. $\alpha_{opt}$.
\end{proof}

We characterize the behaviour of finite constellation FDMA under strong-interference and weak-interference in the following two subsections.

\subsection{Finite Constellation FDMA in Strong-Interference Channel}
For $|h_{12}|$=$|h_{21}|$=$1$, it is easy to see from (\ref{eqnset13c}), (\ref{eqnset16a}) and (\ref{eqnset16b}) that the FDMA rate curve using Gaussian alphabet will touch the capacity curve at $\alpha=\alpha_{opt}=\frac{P_1}{P_1+P_2}$. But with finite constellation, it is not clear from (\ref{eqnset12c}) and (\ref{eqnset15}) whether, at $\alpha_{opt}$, the FDMA rate point will lie on the CC capacity curve or not. So, we need to plot it for some cases and observe the behaviour.
\begin{figure}[htbp]
\subfigure[$W=6$ $Hz$] {\label{fig:7b}\includegraphics[totalheight=2.5in,width=3.5in]{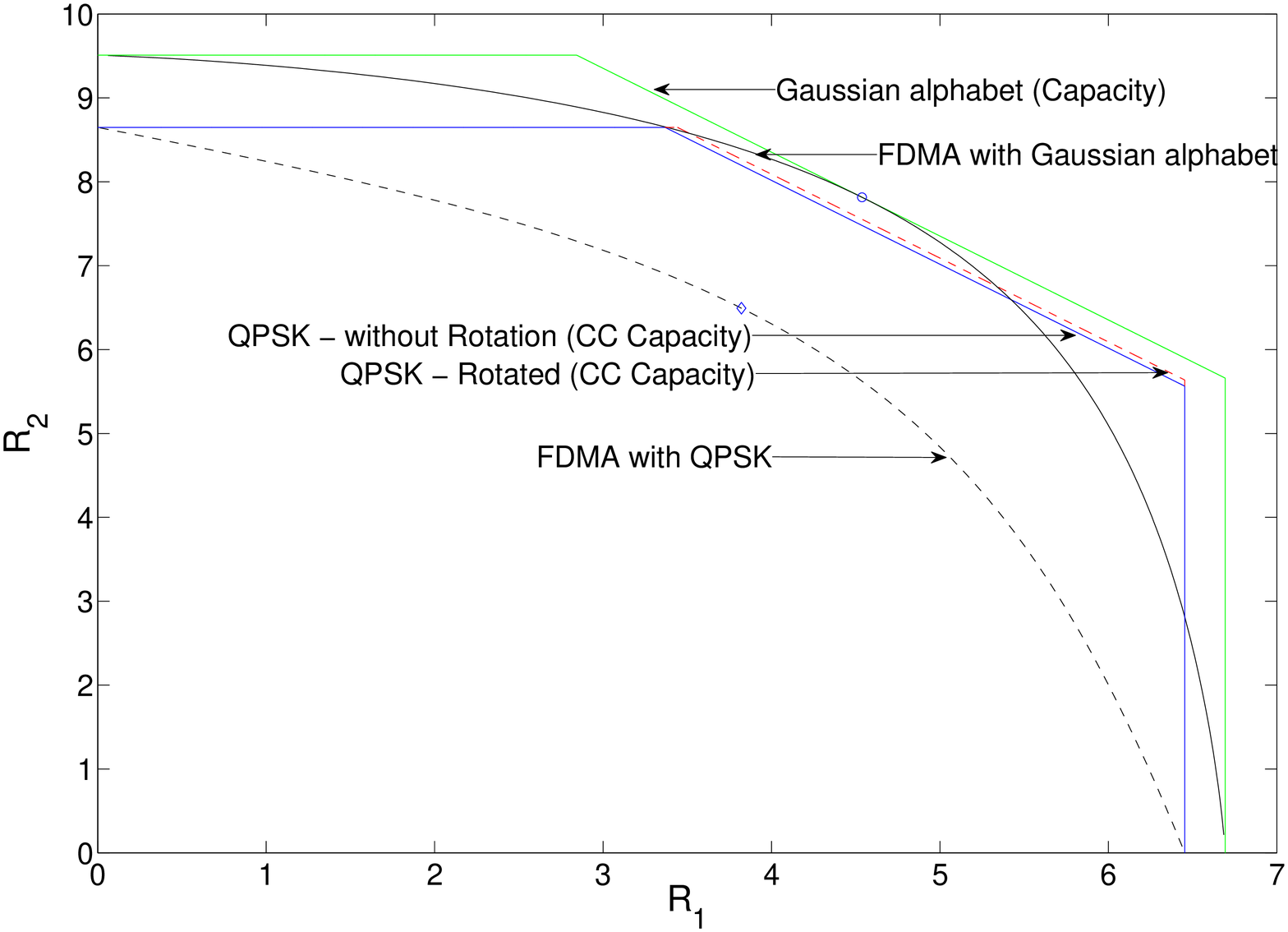}}
\subfigure[$W=2$ $Hz$] {\label{fig:7a}\includegraphics[totalheight=2.5in,width=3.5in]{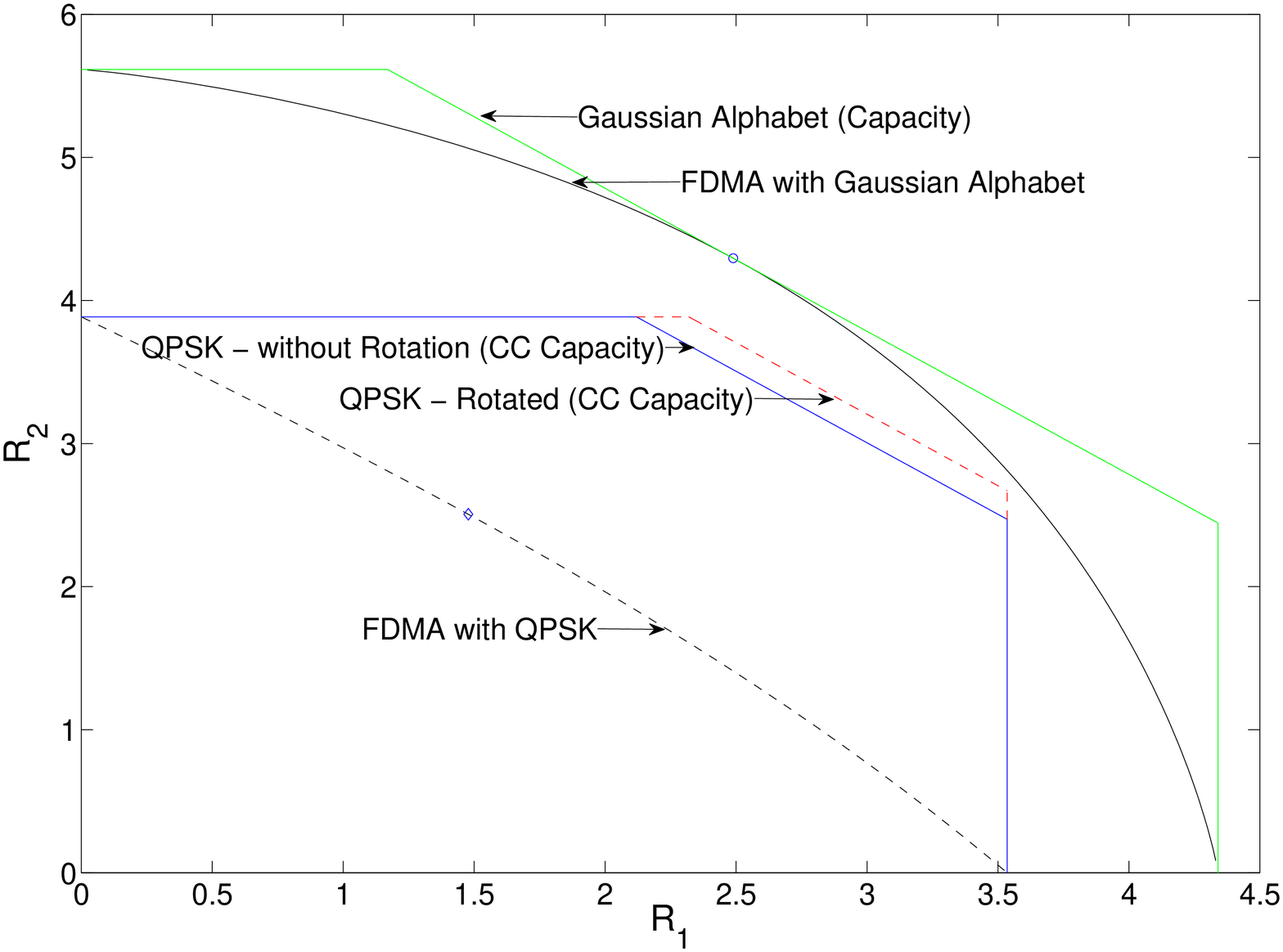}}\\
\caption{FDMA and Capacity Region for QPSK pair at $P_1$=$7$ Watt (=$8.45dB$), $P_2$=$12$ Watt (=$10.79dB$), $h_{12}$=$1\angle10^\circ$, $h_{21}$=$1\angle20^\circ$.}
\label{fig:fig7}
\end{figure}
Rate pairs achieved by FDMA with Gaussian alphabets and FDMA with QPSK alphabets are shown in Fig. \ref{fig:fig7}. Fig. \ref{fig:fig7} represents a case when $|h_{12}|$=$|h_{21}|$=$1$. Since rotation offers increase in the CC capacity, from now on, we consider only the rotated version of the signal set. As seen in Fig. \ref{fig:fig7}, the FDMA rate curve does not touch the CC capacity curve (rotated version) and it moves away from it with decreasing $W$. We can consider, without loss of generality,  the power constraint for User-$i$, for the full bandwidth case, as $\frac{P_i}{W}$ ($i$=$1,2$) and the noise variances as 1 by dividing (\ref{eqnset9}) by $\sqrt W$ and similarly for the FDMA case we take the power constraints to be $\frac{P_i}{W_i}$ ($i$=1,2). The same effect of decreasing $W$ is observed by increasing both $P_1$ and $P_2$ with the same factor by which $W$ is decreased. Note that $\alpha_{opt}$ remains the same when $P_1$ and $P_2$ are increased by the same factor. The reason why the FDMA rate curve goes away from the CC capacity curve by increasing both $P_1$ and $P_2$ by the same factor is given below.

$I_W\left(\sqrt{P_1}X_1,\sqrt{P_2}X_2;Y_1\right)$ and $I_W\left(\sqrt{P_1}X_1,\sqrt{P_2}X_2;Y_2\right)$, the CC capacities of the $16$-point constellations ${\cal S}_{sum_1}$ and ${\cal S}_{sum_2}$ respectively, both of which have an effective average power of $\frac {(P_1+P_2)}{W}$, have to saturate at $4$ bits while both $I_{W_1}\left(\sqrt{P_1}X_1;Y_1|\sqrt{P_2}X_2\right)$ and $I_{W_2}\left(\sqrt{P_2}X_2;Y_2|\sqrt{P_1}X_1\right)$, the CC capacities of $4$-point constellations which also have effective average powers of $\frac {(P_1+P_2)}{W}$ (as they are evaluated at $\alpha=\alpha_{opt}$), have to saturate at $2$ bits when $P_1$ and $P_2$ are increased by the same factor. So, $I_W\left(\sqrt{P_1}X_1,\sqrt{P_2}X_2;Y_1\right)$ and $I_W\left(\sqrt{P_1}X_1,\sqrt{P_2}X_2;Y_2\right)$ increase at a faster rate than $I_{W_1}\left(\sqrt{P_1}X_1;Y_1|\sqrt{P_2}X_2\right)$ and $I_{W_2}\left(\sqrt{P_2}X_2;Y_2|\sqrt{P_1}X_1\right)$. Hence, the difference, normalized with respect to $W$,
\begin{align}
\nonumber 
&\frac{1}{W}\left[\min\left\{WI_W\left(\sqrt{P_1}X_1,\sqrt{P_2}X_2;Y_1\right),\hspace{30cm}\right\}\right]\\
\nonumber
&\hspace{-30cm} \left \{ \hspace{32cm} WI_W\left(\sqrt{P_1}X_1,\sqrt{P_2}X_2;Y_2\right)\right\}-\\
\nonumber
&\hspace{1cm}\left(W_1I_{W_1}\left(\sqrt{P_1}X_1;Y_1|\sqrt{P_2}X_2\right)+ \hspace{30cm} \right)\\
\nonumber
&\hspace{-30cm}\left[ \left(\hspace{31cm} W_2I_{W_2}\left(\sqrt{P_2}X_2;Y_2|\sqrt{P_1}X_1\right)\right)\right]
\end{align}
evaluated at $\alpha_{opt}$, increases by increasing ($P_1$,$P_2$) by the same factor or decreasing $W$. 

The argument with regards to the constellation-constrained FDMA rate curve moving away from the CC capacity curve with decrease in $W$ holds good for constellations with arbitrary size and arbitrary complex values of $h_{12}$ and $h_{21}$, with $|h_{12}|=|h_{12}|=1$. Hence, at a given finite $W$, for the finite constellation case, the FDMA rate curve, under constellation constraints, does not touch the CC capacity curve. But the difference between the optimum FDMA sum-rate and the CC sum-capacity, for a given value of channel gains, will depend on the constellation size.

When either $|h_{12}|$=$1$ and $|h_{21}|>1$ or $|h_{12}|>1$ and $|h_{21}|$=$1$, it is easily seen from (\ref{eqnset13c}), (\ref{eqnset16a}) and (\ref{eqnset16b}), that for the Gaussian alphabet case, the FDMA rate curve will touch the capacity curve at $\alpha=\alpha_{opt}=\frac{P_1}{P_1+P_2}$. But, for the finite alphabet case, it is not clear again from (\ref{eqnset12c}) and (\ref{eqnset15}) whether, at $\alpha_{opt}$, the FDMA rate point will lie on the CC capacity curve or not. Fig. \ref{fig:fig10} is representative of the case when $|h_{12}|$=$1$ and $|h_{21}|>1$. In Fig. \ref{fig:fig10}, the FDMA rate curve with constellation constraints strictly lies within the CC capacity curve. The behaviour with decreasing $W$ is the same as for the case when $|h_{12}|$=$1$ and $|h_{21}|$=$1$. The reason for this is the same as stated for $|h_{12}|$=$1$ and $|h_{21}|$=$1$ except that only one of the sum-constellations (${\cal S}_{sum_1}$,${\cal S}_{sum_2}$) will have an average power of $\frac{P_1+P_2}{W}$ and that will dominate the CC capacity. Hence, for $|h_{12}|$=$1$ and $|h_{21}|>1$, under constellation constraints, the FDMA rate curve lies strictly within the CC capacity curve. The results for $|h_{12}|$=$1$ and $|h_{21}|>1$ are applicable to $|h_{12}|>1$ and $|h_{21}|$=$1$ also.
\begin{figure}
\subfigure[$W=6$ $Hz$] {\label{fig:10a} \includegraphics[totalheight=2.5in,width=3.5in]{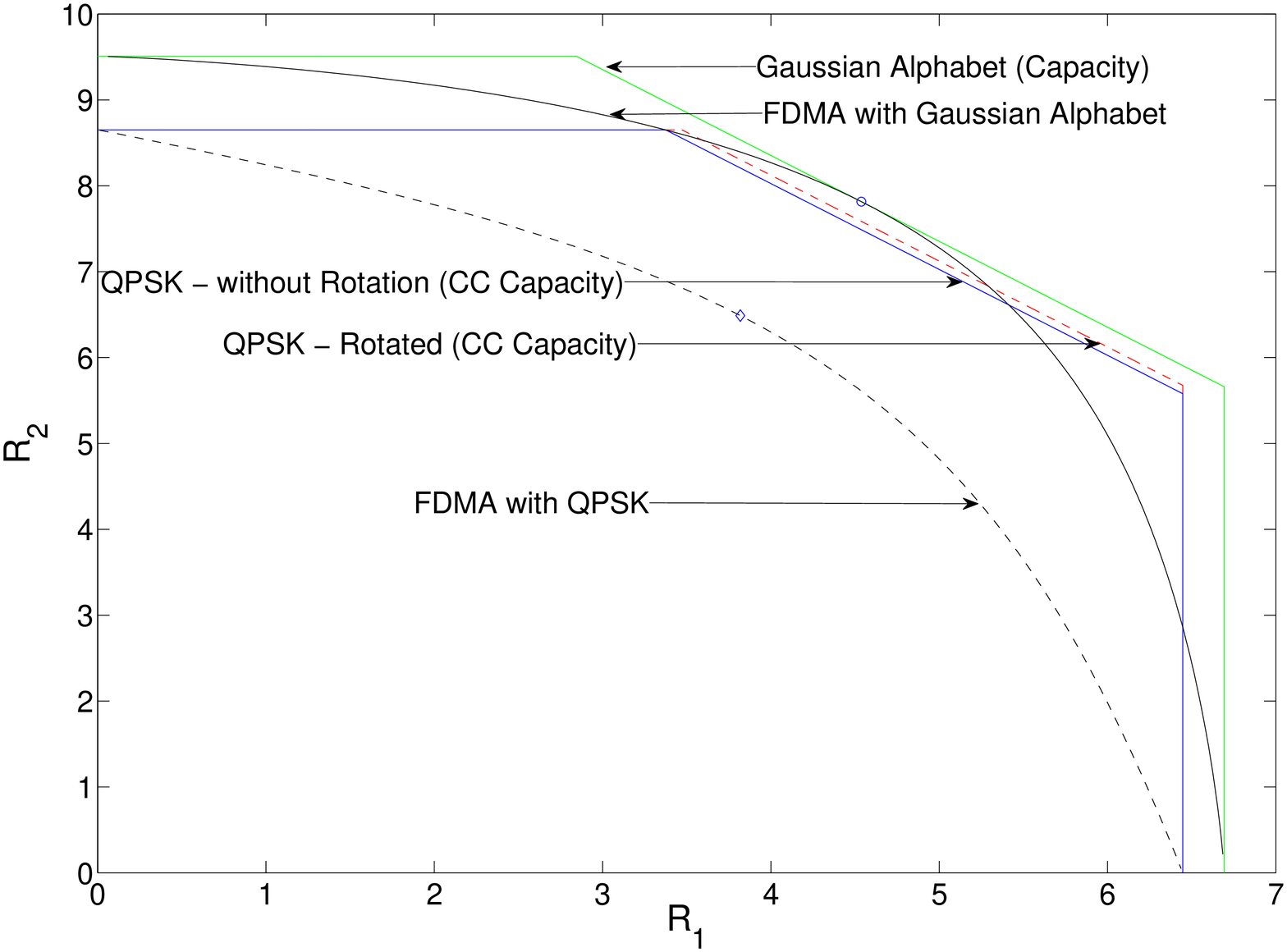}}
\subfigure[$W=2$ $Hz$] {\label{fig:10b} \includegraphics[totalheight=2.5in,width=3.5in]{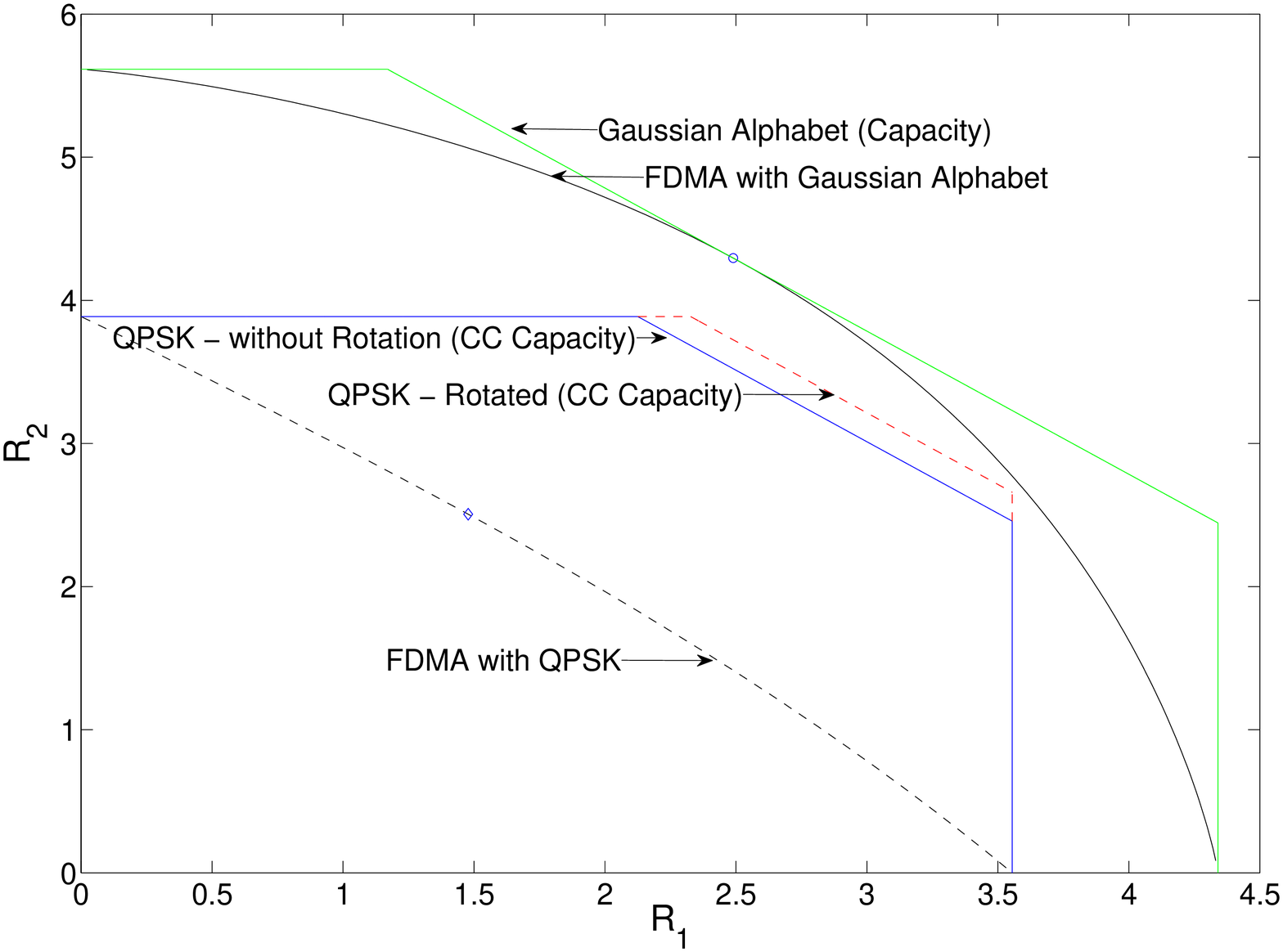}}
\caption{FDMA and capacity curves for QPSK pair at $P_1$=$7$ Watt (=$8.45dB$), $P_2$=$12$ Watt (=$10.79dB$), $h_{12}$=$1\angle10^\circ$, $h_{21}$=$1.1\angle20^\circ$.}
\label{fig:fig10}
\end{figure}

When $|h_{12}|>1$ and $|h_{21}|>1$, as represented by Fig. \ref{fig:fig11}, the FDMA rate curve with Gaussian alphabet doesn't touch the capacity curve (as indicated in \cite{GaY}) which is obvious from (\ref{eqnset13c}), (\ref{eqnset16a}) and (\ref{eqnset16b}). For the finite constellation case too, the FDMA rate curve doesn't touch the CC capacity curve which is also implied by the result that, at $|h_{12}|=1$ and $|h_{21}|=1$, the FDMA rate curve doesn't touch the CC capacity curve.
\begin{figure}
\includegraphics[totalheight=2.5in,width=3.5in]{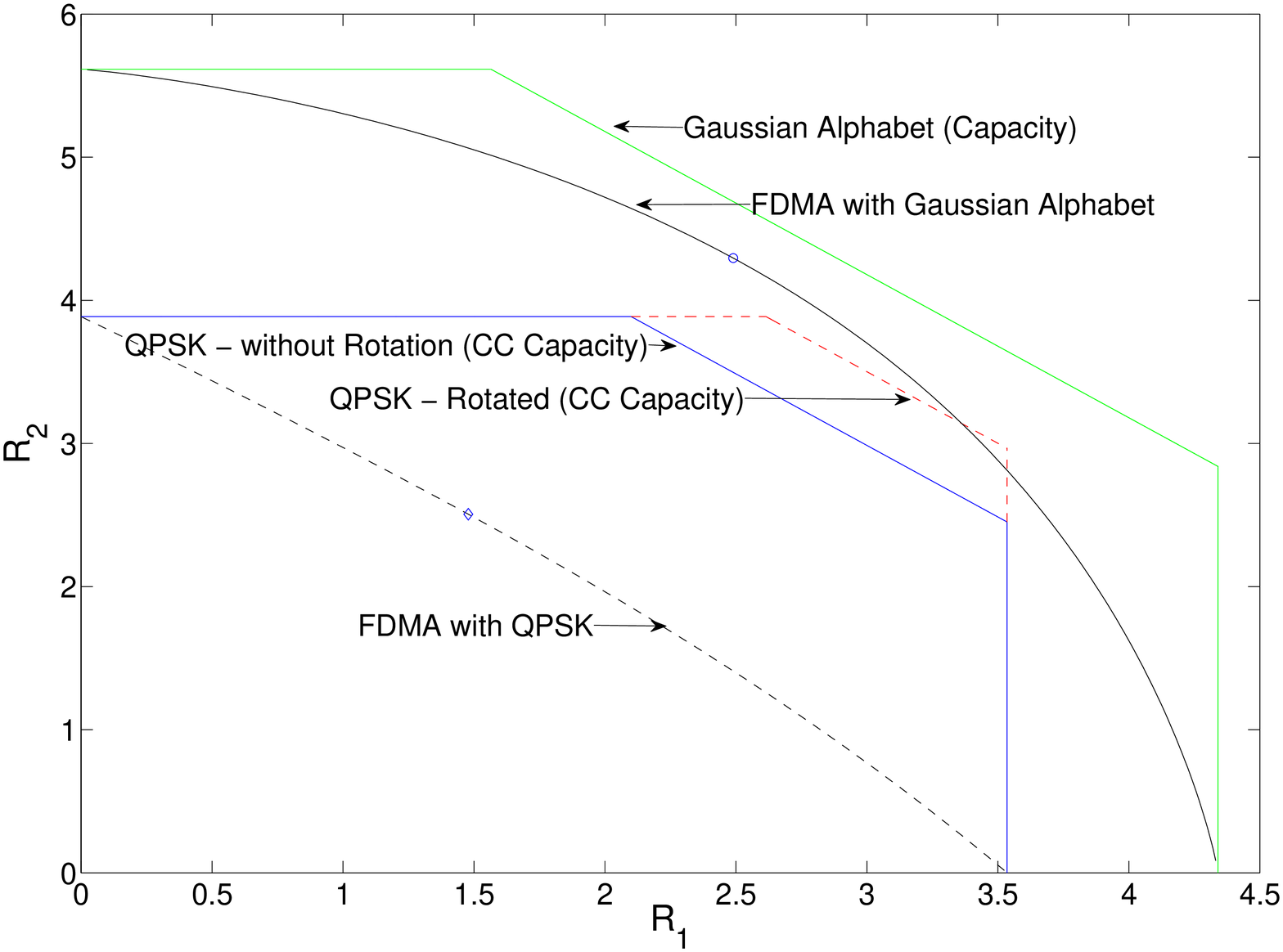}
\caption{FDMA and capacity curves for QPSK pair at $P_1$=$7$ Watt (=$8.45dB$), $P_2$=$12$ Watt (=$10.79dB$), $h_{12}$=$1.2\angle10^\circ$, $h_{21}$=$1.2\angle20^\circ$, $W$=$2$ $Hz$.}
\label{fig:fig11}
\end{figure}

Hence, when $|h_{12}|$=$|h_{21}|$=$1$, $|h_{12}|$=$1$ and $|h_{21}|>1$, and, $|h_{12}|>1$ and $|h_{21}|=1$, the Gaussian alphabet FDMA rate curve will touch the capacity curve while the finite constellation FDMA rate curve will never touch the CC capacity curve in the strong-interference regime.

%
\subsection{Finite Constellation FDMA in Weak-Interference Channel}
When either $|h_{12}|$ or $|h_{21}|$ or both are less than $1$, (\ref{eqnset12a})-(\ref{eqnset13c}) and (\ref{eqnset14a})-(\ref{eqnset16b}) are just inner bounds (i.e. achievable regions). From (\ref{eqnset12a})-(\ref{eqnset12c}), it is seen that the simultaneous-decoding inner-bound for the finite constellation case is enlarged by relative rotation of the finite input constellations. It is clear from (\ref{eqnset13c}), (\ref{eqnset16a}) and (\ref{eqnset16b}) that, for the Gaussian alphabet case, when $|h_{12}|$ or $|h_{21}|$ or both are less than $1$, the FDMA inner-bound, at $\alpha_{opt}$, is always better than the simultaneous-decoding inner-bound.
\begin{figure}[htbp]
\centering
\includegraphics[totalheight=2.5in,width=3.5in]{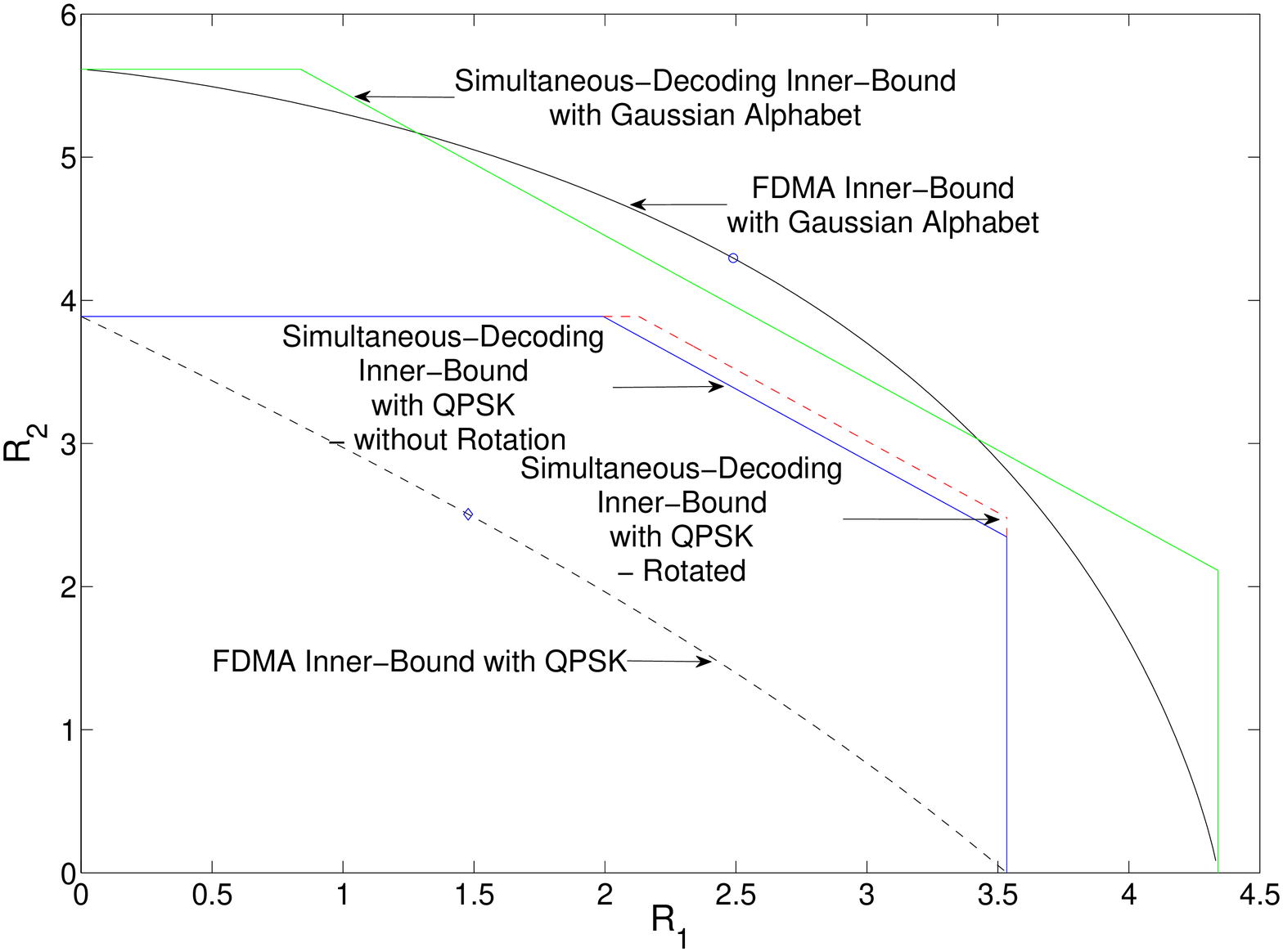}
\caption{FDMA inner-bound and simultaneous-decoding inner-bound for QPSK pair at $P_1$=$7$ Watt (=$8.45dB$), $P_2$=$12$ Watt (=$10.79dB$), $h_{12}$=$1\angle10^\circ$, $h_{21}$=$0.9\angle20^\circ$, $W$=$2$ $Hz$.}
\label{fig:fig8}
\end{figure}
One interesting observation that can be made from Fig. \ref{fig:fig8} is that, for the finite constellation case, the simultaneous-decoding inner-bound still remains strictly better than the FDMA inner-bound. Hence, under weak-interference, when  $|h_{12}|$ and $|h_{21}|$ are close to $1$, the simultaneous-decoding inner-bound outperforms the FDMA inner-bound, at $\alpha_{opt}$, for the finite constellation case, unlike the Gaussian alphabet case. However, under constellation constraints, the values of cross channel gains at which the FDMA inner-bound, at $\alpha_{opt}$, outperforms the simultaneous-decoding inner-bound depends on the constellations used. One instance of the FDMA inner-bound, at $\alpha_{opt}$, outperforming the simultaneous-decoding inner-bound, under constellation constraints, is shown in Fig. \ref{fig:fig9}.
\begin{figure}
\includegraphics[totalheight=2.5in,width=3.5in]{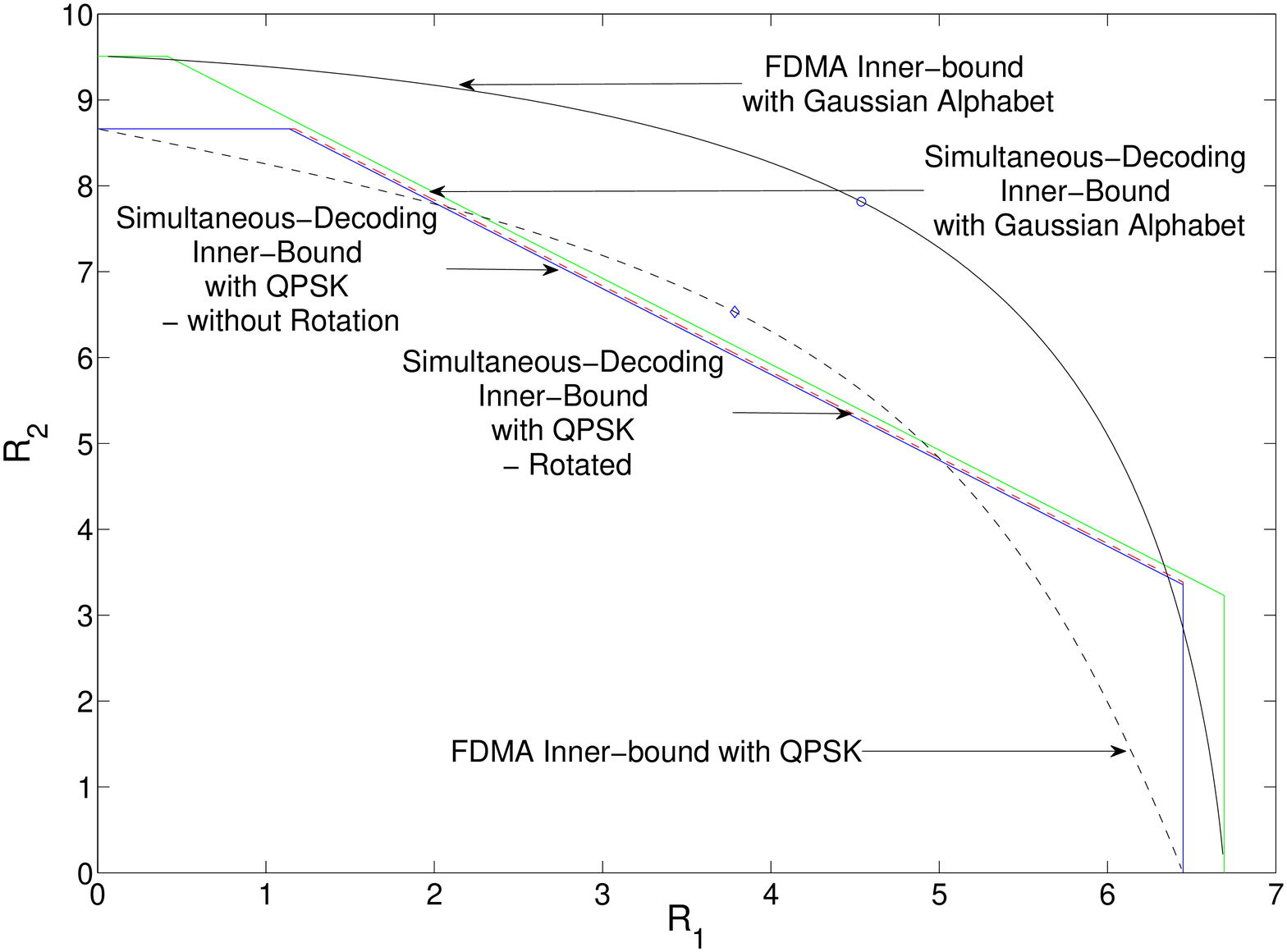}
\caption{FDMA inner-bound and simultaneous-decoding inner-bound for QPSK pair at $P_1$=$7$ Watt (=$8.45dB$), $P_2$=$12$ Watt (=$10.79dB$), $h_{12}$=$1\angle10^\circ$, $h_{21}$=$0.7\angle20^\circ$, $W$=$6$ $Hz$.}
\label{fig:fig9}
\end{figure}

\section{DISCUSSION}
We showed that throughout the strong-interference regime, with finite constellation, the FDMA rate curve never touches the CC capacity curve while for the Gaussian alphabet case, the FDMA rate curve touches the capacity curve for some portion of the strong-interference regime. This is another instance of what holds good for the Gaussian alphabet case need not hold good when finite input constellations are employed (for GMAC such results have already been shown). An interesting direction of future work lies in the weak-interference regime. For some portion of the weak-interference regime, with a symmetric channel and equal powers for both the users, using Gaussian alphabets, the inner-bound obtained from orthogonal signaling is better than the inner-bound obtained from treating interference as noise \cite{ETW}. It would be interesting to see what happens when finite input constellations are used in such a case. 

An important direction to pursue is to develop non-orthogonal multiple access schemes for interference channels which exploit the enlarged portion of the CC capacity and operate above the FDMA rate curve.

\section*{Acknowledgement}
The authors wish to thank T. Damodaram Bavirisetti for the useful discussions on the proofs to the theorems. This work was supported  partly by the DRDO-IISc program on Advanced Research in Mathematical Engineering through a research grant as well as the INAE Chair Professorship grant to B.~S.~Rajan.

\end{document}